\newif\ificml
\newif\ifneurips
\newif\ificlr
\newif\ifusenix
\usenixfalse 
\newif\ifanonymous
\anonymousfalse 
\newif\iffull
\fulltrue
\PassOptionsToPackage{dvipsnames}{xcolor}
\documentclass[letterpaper,twocolumn,10pt]{article}
\usepackage{usenix}
\ifusenix
\usepackage[available]{usenixbadges}
\fi
\ificlr
\usepackage{iclr2026_conference,times}
\fi

\ifneurips
\usepackage[dblblindworkshop,final]{neurips_2025}
\workshoptitle{Regulatable Machine Learning}
\fi

\ificml
\usepackage{icml2026}
\fi

\usepackage[utf8]{inputenc} 
\usepackage[T1]{fontenc}    
\usepackage{hyperref}       
\usepackage{url}            
\usepackage{booktabs}       
\usepackage{amsfonts}       
\usepackage{nicefrac}       
\usepackage{microtype}      
\usepackage{xcolor}         
\usepackage{amsmath}        
\usepackage{amsthm}         
\usepackage{xspace}
\usepackage{algorithm}
\usepackage[noend]{algpseudocode}
\usepackage{comment}
\usepackage{graphicx}
\usepackage{multirow}
\usepackage{pifont}
\usepackage{makecell}
\usepackage{tablefootnote}
\usepackage{todonotes}
\usepackage[justification=centering]{caption}
\newtheorem{lemma}{Lemma}
\newtheorem{cor}{Corollary}
\newtheorem{defn}{Definition}
\newtheorem{thm}{Theorem}

\newif\ifdraft
\draftfalse  

\ifdraft
\newcommand{\olive}[1]{{\color{violet}{OLIVE: #1}}}
\newcommand{\carter}[1]{{\color{blue}{CARTER: #1}}}
\newcommand{\lisa}[1]{{\color{cyan}{Lisa: #1}}}
\newcommand{\akira}[1]{[{\color{red!50} \textsf{Akira:} #1}]}
\newcommand{\np}[1]{{\color{orange}{Nicolas: #1}}}
\newcommand{\jonas}[1]{{\color{teal}{jonas: #1}}}
\else
\newcommand{\olive}[1]{}
\newcommand{\carter}[1]{}
\newcommand{\lisa}[1]{}
\newcommand{\akira}[1]{}
\newcommand{\np}[1]{}
\newcommand{\jonas}[1]{}
\fi

\ificml
\icmltitlerunning{Certified in Theory, Broken in Practice: Assumption Gaps in Cryptographic Model Certification}
\else
\title{Certified in Theory, Broken in Practice: \\Assumption Gaps in
  Cryptographic Model Certification}
\fi

\makeatletter
\newcommand{\samethanks}[1]{%
  \hyperlink{Hfootnote.#1}{%
    \textsuperscript{\normalfont\@fnsymbol{#1}}%
  }%
}
\makeatother

\ifanonymous
\else
\author{%
  Carter Luck\thanks{Equal contribution}\\
  University of Massachusetts Amherst 
  \and
  Olive Franzese-McLaughlin\samethanks{1} \\
  Vector Institute \& University of Toronto 
  \and
  Elisaweta Masserova\samethanks{1} \\
  Carnegie Mellon University
  \and
  Akira Takahashi \\
  J.P.Morgan AI Research \& AlgoCRYPT CoE 
  \and
  Antigoni Polychroniadou \\
  J.P.Morgan AI Research \& AlgoCRYPT CoE 
  \ifneurips
  \else
  \and
  Nicolas Papernot \\
  Vector Institute \& University of Toronto 
  \fi
}

\fi

\newcommand{\class}{\mathsf{class}}
\newcommand{\thr}{\mathsf{thr}}
\newcommand{\attr}{\mathsf{attr}}

\newcommand{\cur}{\mathsf{cur}}
\newcommand{\troot}{\mathsf{root}}
\newcommand{\tleft}{\mathsf{left}}
\newcommand{\tright}{\mathsf{right}}

\newcommand{\ba}{\mathbf{a}}

\newcommand{\ZKP}{\mathsf{ZKP}}
\newcommand{\rel}{\ensuremath{\mathcal{R}}}

\newcommand{\cP}{\mathcal{P}}
\newcommand{\cV}{\mathcal{V}}

\newcommand{\cA}{\mathcal{A}}
\newcommand{\cS}{\mathcal{S}}
\newcommand{\cE}{\mathcal{E}}
\newcommand{\stm}{\mathsf{x}}

\newcommand{\wit}{\mathsf{w}}

\newcommand{\view}{\mathsf{view}}
\newcommand{\secp}{\lambda}
\newcommand{\interact}[2]{\ensuremath{\langle #1,#2\rangle}}
\newcommand{\probrv}[2]{\ensuremath{\Pr\left[\begin{aligned}#1\end{aligned}\,:\, \begin{aligned}#2\end{aligned}\right]}}

\newcommand{\probcond}[2]{\ensuremath{\Pr\left[\begin{aligned}#1\end{aligned}\,\middle|\, \begin{aligned}#2\end{aligned}\right]}}

\newcommand{\COM}{\mathsf{COM}}
\newcommand{\commit}{\mathsf{Commit}}

\newcommand{\com}{\mathsf{com}}

\newcommand{\settrain}{S_\text{train}}
\newcommand{\setaudit}{S_\text{audit}}
\newcommand{\settest}{S_\text{test}}

\newcommand{\train}{\mathsf{Train}}
\newcommand{\audit}{\mathsf{Audit}}
\newcommand{\prove}{\mathsf{Prove}}
\newcommand{\ext}{\mathsf{Ext}}
\newcommand{\simul}{\mathsf{Sim}}

\newcommand{\pext}{p_\text{ext}}
\newcommand{\pacc}{p_\text{acc}}

\newcommand{\auditcsp}{\Pi_\text{csp}}

\newcommand{\querysp}{Q}

\newcommand{\dpemp}{\Delta_\mathsf{dp}}
\newcommand{\dptrue}{\hat{\Delta}_\mathsf{dp}}

\newcommand{\gen}{\mathsf{Gen}}

\newcommand{\cmark}{{\color{green!50!black}\ding{51}}}  
\newcommand{\xmark}{{\color{red!70!black}\ding{55}}}    
\newcommand{\pmark}{{\color{orange!85!black}\ding{115}}}  

\newcommand{\myparagraph}[1]
{\noindent{\bf #1}}

\begin{document}
\ificml
\twocolumn[
\icmltitle{Data Forging Attacks on \\ Cryptographic Model Certification}


\icmlsetsymbol{equal}{*}

\begin{icmlauthorlist}
\icmlauthor{Firstname1 Lastname1}{equal,yyy}
\icmlauthor{Firstname2 Lastname2}{equal,yyy,comp}
\icmlauthor{Firstname3 Lastname3}{comp}
\icmlauthor{Firstname4 Lastname4}{sch}
\icmlauthor{Firstname5 Lastname5}{yyy}
\icmlauthor{Firstname6 Lastname6}{sch,yyy,comp}
\icmlauthor{Firstname7 Lastname7}{comp}
\icmlauthor{Firstname8 Lastname8}{sch}
\icmlauthor{Firstname8 Lastname8}{yyy,comp}
\end{icmlauthorlist}

\icmlaffiliation{yyy}{Department of XXX, University of YYY, Location, Country}
\icmlaffiliation{comp}{Company Name, Location, Country}
\icmlaffiliation{sch}{School of ZZZ, Institute of WWW, Location, Country}

\icmlcorrespondingauthor{Firstname1 Lastname1}{first1.last1@xxx.edu}
\icmlcorrespondingauthor{Firstname2 Lastname2}{first2.last2@www.uk}

\icmlkeywords{Machine Learning, ICML}

\vskip 0.3in
]



\printAffiliationsAndNotice{}  
\else 
\maketitle
\fi 

\begin{abstract}

Privacy-preserving machine learning auditing protocols allow auditors to assess models for properties such as accuracy or fairness, without revealing their internals or training data. This makes them especially attractive for auditing models deployed in sensitive domains such as healthcare or finance.
For these protocols to be meaningful in real-world audit settings, though, their guarantees must reflect how the model will behave once deployed, rather than merely certifying its behavior during an audit. Existing security definitions often miss this mark: most certify model behavior only on a \emph{fixed audit dataset}, without ensuring that the same guarantees \emph{generalize} to other datasets drawn from the same distribution.
As we show, this gap allows a model provider to attack many cryptographic model certification (CMC) schemes built on secure zero knowledge proofs (ZKP) by carefully engineering training data, resulting in models that exhibit benign behavior during an audit, but pathological behavior in practice. For example, we empirically demonstrate that an attacker can certify that a model achieves over 99\% accuracy on an audit dataset, but less than 30\% accuracy on fresh samples from the same distribution.

To address this gap, we formalize rigorous cryptographic security notions tailored to CMC frameworks, introduce a generic protocol template, and prove that it satisfies these requirements. Our results thus offer both cautionary evidence about existing approaches and constructive guidance for designing secure, privacy-preserving ML auditing protocols.

\end{abstract}

\section{Introduction}
\looseness=-1
Certifiable, privacy-preserving machine learning aims to formally prove desired properties of the model while keeping model parameters and training data confidential~\cite{ZhangFZS20,LiuXZ21,shamsabadi2022profitt}. In this context, the typical certification lifecycle follows a sequence in which the model provider first trains the model, then an auditor evaluates it according to desired criteria, and---after passing the audit---the certified model is deployed.\footnote{Some works require continuous auditing during deployment instead of a single audit pre-deployment; see Table~\ref{tab:ppml-landscape}.} Crucially, the trustworthiness of this certification does not rely on trusting the model provider or a specific machine learning algorithm. Instead, it is ensured by cryptographic mechanisms -- such as commitments and zero-knowledge proofs~\cite{GoldwasserMR85} -- which guarantee that the auditor’s evaluation is performed correctly on the claimed model. 
The usage of cryptographic techniques allows the auditor to not only verify whether the model has the intended property, but do so while keeping the model internals and training data private. 

However, as we observe, the guarantees that \emph{cryptographic model certifications} (CMC) provide often fail to reflect the model’s behavior beyond the specific setting in which certification is performed. In this paper, we show that such audit-specific guarantees risk creating a false sense of security: by themselves, they do not ensure that the certified properties will continue to hold once the model is deployed and applied to \emph{fresh data}, even when this data is drawn from the same distribution as the audit dataset.
We demonstrate that this is not merely a theoretical concern.

\begin{table*}[t]
  \centering
  \small
  \caption{Analysis of vulnerabilities to data-forging attacks in privacy-preserving ML audits.\\ \cmark = supported; \pmark = conditional; \xmark = not supported.}
  \label{tab:ppml-landscape}
  \setlength{\tabcolsep}{6pt}
  \begin{tabular}{lcccc c c}
    \toprule
    \multirow{2}{*}{Work} &
    \multicolumn{4}{c}{\textbf{Certified property}} &
    \makecell{\textbf{Resilience to}\\\textbf{data-forging}} &
    \makecell{\textbf{Continuous}\\\textbf{verification}}\\
    \cmidrule(lr){2-5}
         & Acc. & Group Fair & Indv. Fair & Diff. Priv.  \\
    \midrule
Zhang et al.~\cite{ZhangFZS20}   & \cmark & \xmark & \xmark & \xmark & \pmark\ (pd)  & \xmark\\
    Shamsabadi et al.~\cite{shamsabadi2022profitt}   & \xmark & \cmark  & \xmark & \xmark & \xmark & \xmark\\
    Yadav et al.~\cite{yadav2024fairproof}    & \xmark & \xmark & \cmark & \xmark & \cmark  & \cmark \\
    Liu et al.~\cite{LiuXZ21}     & \cmark & \xmark & \xmark & \xmark & \pmark\ (pd)  & \xmark \\
    Franzese et al.~\cite{franzese2024oath}     & \xmark & \cmark & \xmark  & \xmark & \cmark  & \cmark\\
    Shamsabadi et al.~\cite{shamsabadi2024dpproof} & \xmark & \xmark & \xmark  & \cmark & \xmark  & \xmark\\
    Kang et al.~\cite{kangdnn}         & \cmark & \xmark & \xmark  & \xmark & \cmark  & \xmark\\
    Wang and Hoang~\cite{WangH23}             & \cmark & \xmark & \xmark  & \xmark & \pmark\ (pd)  & \xmark\\
    Bourrée et al.~\cite{p2nia}             & \xmark & \cmark  & \xmark & \xmark  & \xmark & \xmark\\
    \bottomrule
  \end{tabular}
 
  {\footnotesize
Acc. = accuracy; Group/Indv. Fair = group/individual fairness; Diff. Priv.=differential privacy. “Conditional” works lack detail to assess resilience to data-forging, but indicate deployments with public datasets (pd), which would be make the solution vulnerable. Continuous verification means audits must run continuously during deployment (e.g., via clients) rather than once pre-deployment. 
  }
\end{table*}

In more detail, consider a certification workflow, where the audit dataset is drawn from public benchmarks, or derived from distributions that providers can closely approximate. In such scenarios, when the audit dataset is known to the model provider -- even partially -- prior to the start of the audit procedure\footnote{Or, alternatively, if the audit procedure does not require an audit dataset at all, such as in ~\cite{shamsabadi2024dpproof}.}, we describe simple yet powerful attack strategies that can be executed by an adversarial model provider. Our attacks allow the provider to pass an audit (thus enabling deployment) while simultaneously pursuing its own, potentially conflicting, interests. 
For example, while an auditor may seek to verify fairness, the model provider may instead prioritize accuracy—even when accuracy and fairness are in tension. 
As a result of the attack, the audited model passes an audit without exhibiting the intended behavior in practice, and this holds even if strong cryptographic tools are used to ensure that the model during the audit procedure is \emph{the same} as the one used during the deployment.
Our threat model reflects common certification practice (both in the academic literature and in practice). 
As surveyed in Appendix~\ref{sec:case_study}, many prior works in model certification utilize publicly known datasets for audit purposes. 
Moreover, when a company audits a model provided by another company, the audit dataset—whether sourced from the company's synthetic test portfolio or from public benchmarks—is typically disclosed to the model provider, with no mechanism in place to ensure that providers finalize their trained models \emph{before} seeing the audit data.

The key idea in our attacks that allows us to bypass common cryptographic defenses is that we do not alter the training process itself. Instead, we show that knowing the certification procedure enables the model provider to forge training data so that a model honestly trained on it passes the audit, yet still exhibits pathological behavior on other datasets coming from the same distribution as the audit dataset. 
We demonstrate the effectiveness of our attacks against certification procedures for accuracy, fairness, and differential privacy.


We empirically show that our \textbf{data forging attacks} can cause dramatic gaps between audit-time guarantees and true model performance: for instance, in one of our attacks a model can pass an audit requiring 99\% accuracy on the audit dataset, yet achieve only 30\% accuracy on new samples from the same distribution. We show that our attacks remain undetected by straightforward approaches such as statistical tests, e.g., Welch's $t$-test~\cite{welch1947generalization} are performed to check whether the training data and audit data were taken from the same distribution. 

To assess how widespread vulnerability to data forging attacks may be, we conduct a case study covering nine recent papers published at top venues in both security and machine learning. We find that six of these works are at least partially vulnerable\footnote{Several works do not specify the auditing protocol or formal security guarantee in sufficient detail to allow a definitive assessment. However, many explicitly consider deployments based on public audit datasets, under which their approaches are vulnerable.}. We provide    a summary of results in Table~\ref{tab:ppml-landscape}.

Given these vulnerabilities, we introduce a formal foundation for cryptographic model certification that provides rigorous guidance for secure auditing practice.
This foundation includes a security definition ensuring that a model provider passes the audit if and only if the model satisfies the desired property (up to a small quantifiable error) on the underlying data distribution. 
We further formalize an attack game that highlights the gap between certifying a property on a fixed dataset and certifying that the same property generalizes to fresh samples from the distribution. 
Finally, we propose a generic template for achieving secure CMC. 
Our approach is agnostic to the specific property that is being certified and, as we formally prove, guarantees that whenever a model passes an audit, the certified properties will also hold at deployment. The key ingredient is ensuring that the audit is conducted on test data that is \emph{independent} of both the model and its training data. This method might serve as a template for future works to obtain not only efficient, but also secure auditing solutions. Applying this method allows to restore/expand security in some of the case study works. 

\ifneurips
Motivated by these vulnerabilities, we introduce a unified syntax capturing existing auditing schemes and define a formal attack game that highlights the gap between certifying a property on a fixed dataset and certifying that the property generalizes to fresh samples from the distribution.
Given the existence of data-forging attacks in this setting, we underscore the importance of conducting audits on test data that is \emph{independent} of both the model and its training data. This approach may serve as a template for future work to achieve auditing solutions that are not only efficient but also secure.
\fi

In summary, our work advances the study of cryptographic auditing for machine learning by 
\textbf{(i)} identifying simple, yet effective, attack strategies that allow to pass an audit while enabling pathological model behavior at deployment with respect to new datasets from the same distribution as the audit one,  
\textbf{(ii)} empirically demonstrating the effectiveness of our attack against generic certification approaches for four objectives: accuracy, fairness, differentially private training, and statistics for distribution similarity testing,
\textbf{(iii)} conducting a systematic case study of prior certification approaches and corresponding threat models proposed in the literature and identifying several existing approaches that are susceptible to data forging vulnerabilities, and
\textbf{(iv)} introducing formal security definitions tailored to certifiable machine learning and a protocol template that mitigates the attack.
We emphasize that we \emph{do not suggest that prior cryptographic works are broken on a technical level}, rather that the guarantees these works provide deserve closer scrutiny. Our findings comprise strong evidence that secure audit solutions with any of the following properties are unlikely: a) those which utilize known public datasets for test purposes, and b) those that reuse test datasets (if model owner learns a substantial amount of this test dataset during the audit). This evidences the importance of continuous sampling of fresh data for a successful audit infrastructure. We hope that our work will inform the design of future cryptographically secure machine learning audit frameworks.

\section{Threat Model: Cryptographic Model Certification (CMC)}\label{sec:mlback}

We consider two parties: a malicious model provider a.k.a. prover, and an honest-but-curious auditor\footnote{We focus on honest-but-curious auditors for simplicity: it allows us to prove zero-knowledge (Theorem~\ref{thm:auditcsp-general}) from the weakest possible assumption, namely an honest verifier ZK. Our main focus is on the soundness gap caused by malicious model providers, so this choice does not affect the core contributions. This is not an inherent limitation: if the underlying proof system is ZK against fully malicious verifiers (e.g., a non-interactive proof system in the common reference string or random oracle model), then $\auditcsp$ achieves ZK against malicious auditors by the same argument.} a.k.a. verifier. The auditor and model provider are in agreement on an interactive audit protocol wherein both input data (some of which may be \emph{proprietary} data private to the model provider) and the model provider additionally inputs a set of model parameters. The auditor's goal is to assess whether the model has trustworthy properties such as accuracy, fairness, or differential privacy, while maintaining the confidentiality of the model and proprietary data. 
The protocol will output a bit indicating whether the model provider ``passes'' the audit. 

The model provider may submit arbitrary model parameters, and arbitrary data to the protocols requested by the auditor. In some scenarios in the paper, we consider a model provider who can gain access to some or all of the auditor's data (mimicking many real-world scenarios where audit data is reused) prior to the audit. The goal of the model provider is to pass the audit, while exploiting their hidden information to surreptitiously embed the model with nefarious properties.

\paragraph{Cryptographic Model Certification vs. Data Poisoning.}
Attacking cryptographic model certification is related, but distinct from, the problem of data poisoning~\cite{SteinhardtKL17}. Such attacks have traditionally been considered in the context of ML systems trained on user-provided data. Both data poisoning and our attacks on CMC (see~\S\ref{sec:attack-methods}) involve adversarial manipulations of training data. However, their \textbf{threat models differ}, leading to fundamentally different attacks and defenses:

In data poisoning the model provider is the \emph{defender}, and it seeks to learn an accurate model. The \emph{attacker} is an external party who aims to corrupt the learned model by inducing the provider to incorporate an auxiliary training dataset $D_p$ of size no larger than clean dataset $D_c$~\cite{BarrenoNJT10}. In contrast, in CMC the roles are reversed: \textbf{the attacker is the AI/ML service provider itself}, while the defender is an external auditor party. The goal of the service provider is to engineer a model that passes an audit, while violating the certified properties on real-world data. This scenario considers a vastly different set of adversarial capabilities, and imposes additional constraints on the defender:

\textbf{1. A CMC adversary can arbitrarily manipulate both the model and the training data.} This includes both label perturbation and feature-level manipulations. It further encompasses adversaries corrupting all training data, i.e., it \emph{allows for traditional poisoning attacks, but without any restrictions on the poisoning ratio}. In fact, an adversarial model provider might simply chose not to include any clean data at all. Finally, the adversary may even corrupt the model parameters directly. Note that simply asking a model provider to execute a certain protocol (including a data poisoning defense mechanism) is not sufficient – a malicious model provider could choose to simply ignore this instruction.
 
\textbf{2. Both the model and the training data are hidden from the defender.} In contrast to poisoning threat models where the defender has full access to the model and data, the defending party in CMC is an external auditor, who is not allowed to view the model and data as they are typically proprietary. \footnote{CMC can also be applied to auditing openly accessible (e.g., open-source) models, where confidentiality of model weights is not a concern. In this setting the system can be made more efficient: one can use a succinct proof system without the zero-knowledge property to reduce verification overhead, and skip generating cryptographic commitments to model weights and the corresponding commitment-opening proofs. The soundness guarantees of this work apply in this setting as well.} This imposes additional constraints on the defender, who is tasked with assessing the validity of the model while keeping its parameters confidential. This is why known solutions with provably secure guarantees (including ours) crucially rely on cryptographic techniques.

To summarize, a CMC adversary can make use of  classical poisoning techniques, including large-scale perturbations, label corruption, and visible backdoor triggers, to mount its attacks. However, the CMC adversary is strictly stronger than the adversary typically considered in traditional poisoning: It can use \emph{additional} attack strategies that involve, e.g., arbitrarily modifying all data and even directly setting model parameters.  Furthermore, the CMC's goal is fundamentally different: it does not attempt to detect or mitigate poisoned data, but instead ensure that any model that passes the audit must satisfy the audited property, regardless of how it was produced. 

\paragraph{Cryptographic Model Certification and Public Benchmarks in ML.}
The conclusions we draw about requiring fresh data for auditing are semantically related to work on the inadequacy of public benchmarks in machine learning~\cite{zhang2025benchmark,hardt2025emerging}, but those works do not consider cryptographic security. For additional related work and an overview of certifiable ML, see~\S\ref{sec:relatedworkcontinued}. 

In~\S\ref{sec:construction-overview}, our goal is to design a CMC such that -- even under these powerful adversarial capabilities -- if the model passes the audit, then it must genuinely satisfy the audited property (e.g., a fairness or accuracy threshold).

\section{Background}\label{sec:relatedwork}

Consider the following scenario: \textcolor{BlueViolet}{An auditor wishes to verify whether a model utilized by an insurance company to justify claim decisions (approve/deny claim) is accurate on a dataset of the auditor's choosing. At the same time, the company does not want to reveal its model due to concerns about privacy and business competition.} In CMC cryptographic techniques are deployed to reconcile these seemingly conflicting goals. 

\subsection{Zero Knowledge Proofs} 
Among these techniques, the central tool is \emph{zero-knowledge (ZK) proofs}\cite{GoldwasserMR85}, a classical cryptographic primitive, which allows one party (a \emph{prover}) prove a statement $\stm$ to another party (\emph{verifier}) without revealing anything else apart from the validity of this statement. 
Such proofs are constructed for a concrete NP relation $\rel$, which is used to formalize what it means for a statement to be true by specifying the type of evidence (witness $\wit$) that certifies it. The statement $\stm$ is public, the witness $\wit$ is private, and the zk proof checks $(\stm,\wit)\in\rel$, without revealing $\wit$.   
In certifiable ML, such proofs allow \textcolor{cyan}{model provider} (\textcolor{cyan}{prover}) to formally prove that a \textcolor{CarnationPink}{model} (\textcolor{CarnationPink}{witness}) satisfies a \textcolor{ForestGreen}{desired property (e.g., accuracy, fairness, or inference correctness) on a given test dataset} (\textcolor{ForestGreen}{statement}) without learning anything else about the model or the training data.\ More formally:


\begin{defn}[Proof System] An (interactive) \emph{proof system} $\ZKP$ for an NP relation $\rel$ is a tuple of interactive Turing machines $(\cP,\cV)$, where $\cP$ is prover and $\cV$ is verifier.
Let $b\gets\interact{\cP(\wit)}{\cV}(\stm)$ denote the interaction between $\cP$ and $\cV$, where both $\cP$ and $\cV$ take $\stm$ as common inputs, and $\cP$ additionally takes $\wit$ as a private input. At the end of interaction, $\cV$ halts by outputting a binary $b$.
\end{defn}

Proof systems that are used in ML auditing typically require the following security properties: For an NP relation $\rel$, they must provide \emph{completeness} (i.e., if the prover and verifier follow the protocol with input $(\stm,\wit)\in\rel$, the verifier always accepts), \emph{knowledge soundness} (i.e., if the verifier accepts, then it must be that the prover owns a valid witness $\wit$ satisfying the given NP relation w.r.t. statement $\stm$), and \emph{zero knowledge} (i.e., the transcript of the interaction between an honest prover and an honest verifier leaks nothing except that there exists a witness $\wit$ such that $(\stm, \wit) \in \rel$).
We defer formal definitions to \iffull \S\ref{sec:zksecurity}.\fi 


\subsection{Cryptographic Commitments} 
Returning to our example, suppose \textcolor{BlueViolet}{the insurance company has successfully passed an audit and can now deploy its model. How can a customer submitting inference queries be assured that the company continues to use the \emph{certified} model—rather than switching to a different, unverified one? Again, the company still wishes to keep its model private.}
The standard cryptographic tool here is \emph{commitments}, which bind the provider to a single private model during the audit. This prevents “model switching” and ensures that model used in deployment is the same as the one that was certified:

\begin{defn}[Commitment Scheme]
A commitment scheme is an algorithm $\commit$,
which is executed as $\com\gets\commit(m;\rho)$. 
It takes as input a message $m$, a uniformly sampled randomness $\rho$, and returns a commitment $\com$. 
\end{defn}

We require two security properties: \emph{hiding} (i.e., given a commitment $\com$, it leaks nothing about the message $m$), and \emph{binding} (i.e., it is computationally infeasible to open the same $\com$ to distinct $m$ and $m'$). 
We defer formal definitions to \iffull \S\ref{sec:comsecurity}. \fi
Now, auditing may require publishing such a commitment to the model,\footnote{The commitment may be signed by the auditor.} after which the client and insurance company engage in a ZK proof of inference against it\footnote{
In more detail, the insurance company will supply the inference outcome $y$ along with a ZK proof proving that 1) the inference outcome is $y$ has been computed using a private model $M$ and the given client's input, and 2) $M$ is the model hidden behind the previously published commitment that was approved by the auditor.}. Figure~\ref{fig:auditprocess} shows the full  certification workflow, where the audit dataset is revealed to the model provider prior to committing to the model.

\begin{figure*}[t]
    \centering
    \includegraphics[width=0.8\textwidth]{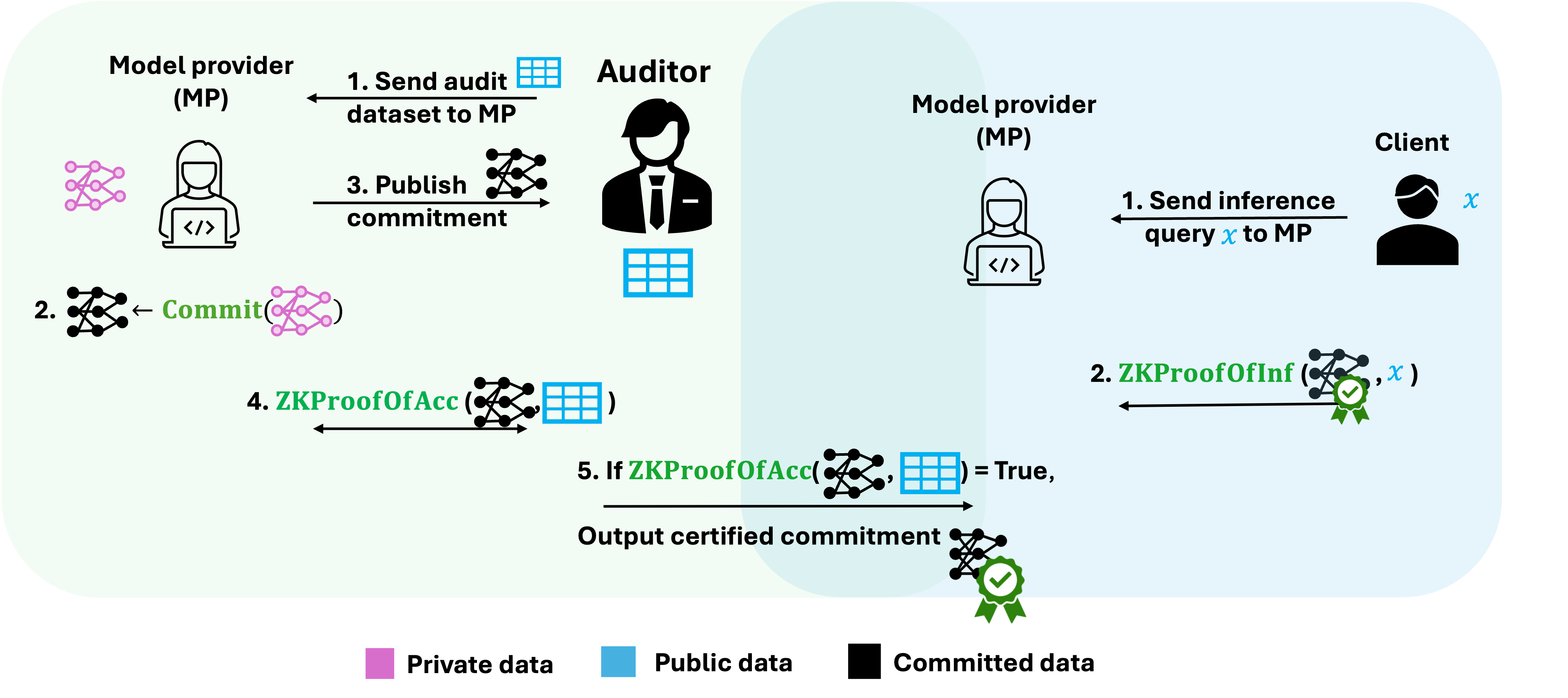}
   \caption{ 
   \footnotesize{Simplified protocol flow for (insecure) ZK-based ML certification. Left: The model provider, after observing the audit dataset, commits to a model and engages with the auditor in a zero-knowledge proof of accuracy (ZKProofOfAcc). If the audit succeeds, the auditor certifies the committed model. Right: For each new inference query, the model provider interacts with the client in a zero-knowledge proof of inference (ZKProofOfInf) protocol, ensuring that the result is consistent with the previously certified commitment.}}
    \label{fig:auditprocess}
\end{figure*}

\subsection{Unifying Syntax for Prior CMC Schemes}\label{sec:previous-cmc-syntax}
We will next discuss the security guarantees of works that address the certification stage (Figure~\ref{fig:auditprocess}, Left)
---namely, proofs of accuracy, fairness, etc., between auditor and model provider. To analyze these systematically, rather than case by case, we abstract away implementation details and introduce a unifying syntax that captures a broad class of existing auditing systems \cite{ZhangFZS20,shamsabadi2022profitt,shamsabadi2024dpproof,LiuXZ21,WangH23,p2nia}. 
To this end, we introduce a \emph{certification predicate} $f$ outputting binary depending whether the input satisfies certain properties. 
Then given a commitment scheme $\commit$ and $f$, we can define a $\ZKP$ for the following relation $\rel$: given $\stm = (\mathsf{com}_h, \mathsf{com}_S, \setaudit)$ as a public statement and $\wit=(h,\rho_h,\rho_S)$ as private witness, $\rel$ outputs 1 iff (1) $f(h, \settrain, \setaudit) = 1$, (2) $h$ is the ML model (or hypothesis) committed in $\mathsf{com}_h$ under randomness $\rho_h$, and (3) $\settrain$ is the training data (which can also include randomness used for training if the training algorithm is randomized) committed in  $\mathsf{com}_S$ under randomness $\rho_S$.
With these building blocks, we define an (insecure) CMC scheme as follows:

\begin{enumerate}
    \item Auditor samples $\setaudit$ from some distribution $\mathcal{D}$ (or uses a public one) and sends $\setaudit$ to the model provider.
    \item Model provider sends cryptographic commitments to its model $\mathsf{com}_h = \commit(h;\rho_h)$ and to the training data $\mathsf{com}_S = \commit(\settrain;\rho_S)$ to the auditor, where $\rho_h$ and $\rho_S$ are sampled uniformly from the randomness space.
    \item They interact to execute $\ZKP$: $b\gets\interact{\cP(h,\settrain)}{\cV}(\mathsf{com}_h, \mathsf{com}_S, \setaudit)$, where model provider plays $\cP$ and the auditor plays $\cV$ and outputs $b$. 
\end{enumerate}

If the output $b$ is $1$, the auditor gets assurance that private model $h$ and training data $\settrain$ satisfy relation $\rel$ without learning anything else about $h$ or $\settrain$. 
We note that, depending on the concrete instantiation of $f$ (see \S~\ref{sec:predicates} for examples), some steps may be omitted.
For example, \cite{ZhangFZS20} verifies accuracy of a model $h$, and their predicate $f$ only needs to check that $h$'s empirical error on $\setaudit$ is below a given threshold. 
Hence, their protocol does not require $\settrain$ and $\mathsf{com}_S$ since they are never used. 
We keep our syntax intentionally general to capture a wide range of prior works. 

\medskip
\noindent\textbf{Security of CMC Schemes}
At first glance, one might expect the security of the above CMC scheme to follow directly from the aforementioned properties of the underlying zero-knowledge proof system $\ZKP$ and commitment scheme $\commit$. 
This is indeed the case for confidentiality, i.e., privacy of the model and training data. 
However, as we show in Section~\ref{sec:attack}, (knowledge) soundness of $\ZKP$ and binding of $\commit$ do \emph{not} necessarily imply the desired notion of soundness: a certified model may satisfy the audited property on the specific audit dataset yet fail to generalize beyond it.
In particular, the soundness of $\ZKP$ is defined only with respect to a fixed statement $\stm=(\mathsf{com}_h,\mathsf{com}_S,\setaudit)$, which binds the proof to a \emph{particular} dataset $\setaudit$. 
Consequently, the proof alone does not ensure good behavior on \emph{fresh} data sampled from the distribution, which is the central goal of auditing. 
To the best of our knowledge, prior CMC work does not formalize appropriate security goals. 
In Section~\ref{sec:construction-overview}, we introduce formal definitions that address it and present a CMC scheme that satisfies them.


\subsection{Example of Certification Predicates}\label{sec:predicates}
We now list concrete examples of predicates $f$ considered in prior works on CMC. 

\medskip
\noindent\textbf{Certifying Accuracy}
To certify accuracy, we consider the empirical error rate as follows:
\begin{align*}
	\hat{\ell}_S(h) &= \frac{1}{n}\sum_{(x,y)\in S}\mathbb{I}(h(x)\neq y)
\end{align*}
where $n = |S|$, $x$ is a data sample, and $y$ is the corresponding label. 
For a public error threshold $t$, the auditor can check whether the model $h$ has empirical error below $t$ on the audit dataset $\setaudit$ by defining the corresponding predicate $f$ as follows:
\begin{align*}
    f(h,\setaudit) = 1  & \iff \hat{\ell}_{\setaudit}(h) \leq t
\end{align*}

\noindent\textbf{Certifying Fairness with Demographic Parity}
Demographic parity \cite{calders2009demographic} is one of the most basic fairness metrics, measuring the difference between the prediction probabilities conditioned on a sensitive attribute.
We consider the empirical parity differences as follows: 
\begin{align*}
	\dpemp(h,\setaudit) &= \left|\frac{1}{n_0}\sum_{x\in S_0}\mathbb{I}(h(x) = 1) - \frac{1}{n_1}\sum_{x\in S_1}\mathbb{I}(h(x) = 1) \right| 
\end{align*}
where $s_x$ denotes the sensitive feature of a data point $x$, $S_0 = \{x \in \setaudit : s_x = 0 \}$, $S_1 = \{x \in \setaudit : s_x = 1 \}$, $n_0 = |S_0|$, and $n_1 = |S_1|$.
With this definition, the model is considered ``fair'' if the parity difference $\dpemp$ is low. 
For a public parity threshold $t$, the auditor can check whether the model $h$ has empirical parity difference below $t$ on the audit dataset $\setaudit$ by defining the corresponding predicate $f$ as follows:
\begin{align*}
    f(h,S_\text{audit}) = 1  & \iff \dpemp(h,\setaudit) \leq t
\end{align*}

\noindent\textbf{Certifying Correct Training Procedure} To check whether a model $h$ is correctly trained according to a specified training procedure $\train$ on a training dataset $\settrain$, we define the predicate $f$ as follows:
\begin{align*}
	f(h,\settrain) = 1  & \iff h = \train(\settrain)	
\end{align*}
where a randomized training procedure can be captured by assuming $\settrain$ contains training randomness. 

\medskip
\noindent\textbf{More Complex Predicates} 
We note that these examples are not exhaustive, and more complex predicates can be constructed by defining an arbitrary predicate $f$ amenable to general-purpose zero-knowledge proofs or by combining multiple predicates.
For instance, one may define $f$ to certify both accuracy \emph{and} correct training procedure simultaneously as follows:
\begin{align*}
	f(h,\settrain, \setaudit) = 1  & \iff \hat{\ell}_{\setaudit}(h) \leq t \, \land\, h = \train(\settrain)	
\end{align*}


\section{Attacking CMC}\label{sec:attack}
Returning to our example, suppose \textcolor{BlueViolet}{the insurance company saves costs by \emph{denying} claims.} 
Intuitively, an accuracy audit with provable guarantees—such as those provided by ZK proof-based systems—and with a sufficiently high threshold (e.g., passing only if accuracy on the auditor’s dataset exceeds 95\%) should prevent the company from deploying a model that unjustifiably denies too many claims.

We show that this intuition is false. Because machine learning is inherently data-dependent, certified properties need not hold once the model is deployed and applied to fresh data, even when drawn from the same distribution. More formally, while prior works certify that 
$$f(h,\settrain, \setaudit)=1$$ 
for some predicate $f$, a given model $h$, training data $\settrain$, and audit dataset $\setaudit$, this does not imply that the stronger property $F$ such that
$$F(h,\settrain)=1 \iff \Pr_{\settest \gets \mathcal{D}}[f(h,\settrain,\settest)=1]>\mathsf{p},$$
where $\mathsf{p}$ is a (non-negligible) probability and $\mathcal{D}$ is a distribution over the entire population $\querysp=\{(x_i,y_i)\}_{i=1}^m$.
The true goal of an audit, however, is precisely such stronger guarantees: an auditor typically seeks to ensure that a model remains fair, accurate, or robust not only on a particular dataset, but also on the unseen datasets it will encounter during deployment.

\begin{figure}[t]
\captionsetup{justification=raggedright,singlelinecheck=false}
    \includegraphics[width=0.49\textwidth]
    {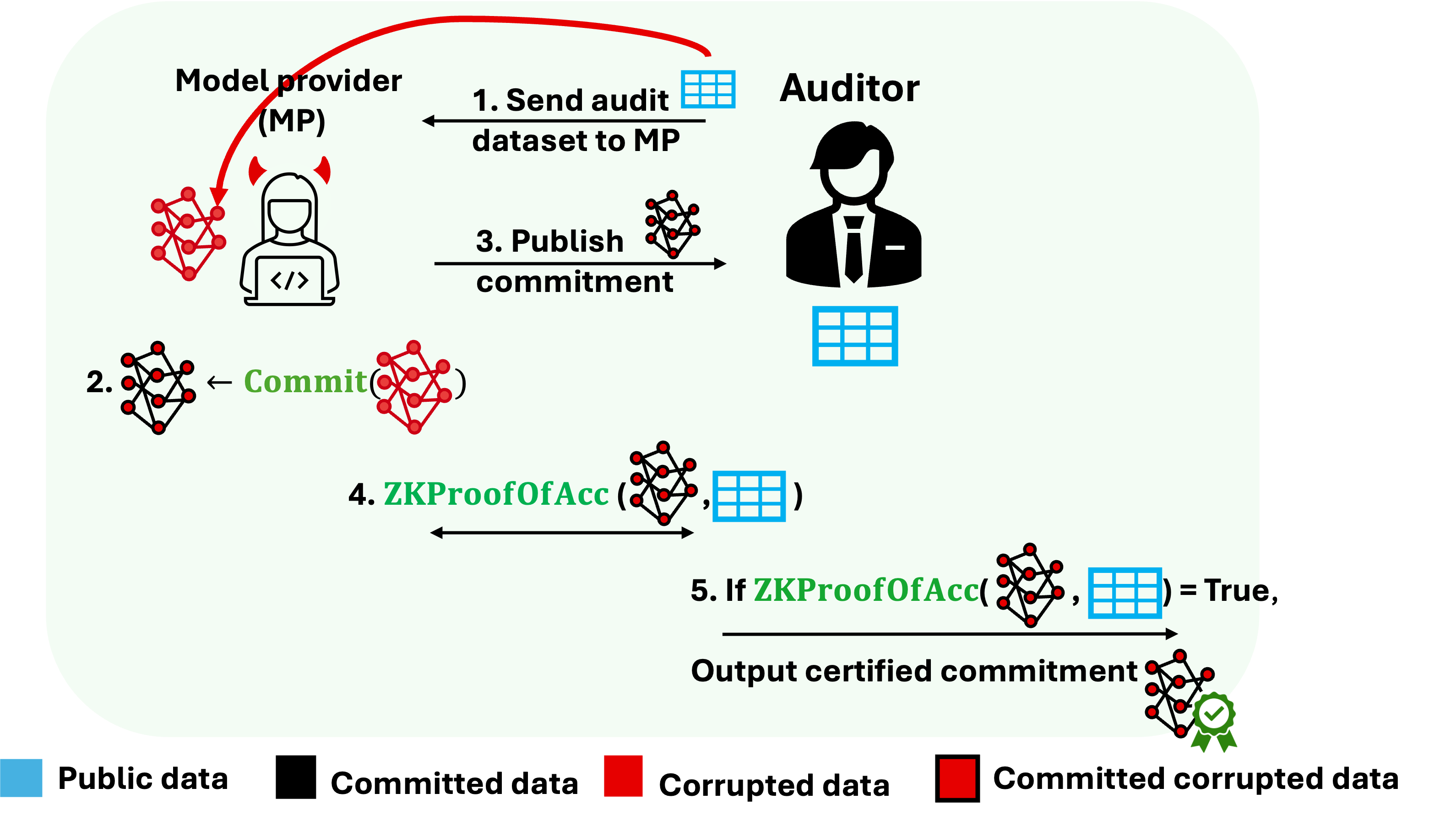}
   \caption{ 
   \footnotesize{Protocol flow during a data forging attack. The model provider, i.e., the attacker, after observing the audit dataset, commits to a model engineered using the knowledge (indicated by red arrow) of this dataset. Then, it engages with the auditor, i.e., the defender, in a zero-knowledge proof of accuracy (ZKProofOfAcc). If the audit succeeds, the auditor certifies the committed corrupted model.}}
    \label{fig:auditprocessmal}
\end{figure}

We show that this gap can be exploited. In particular, if $\setaudit$ is partially known to the model provider before it is required to cryptographically commit to the model, the provider can ensure $f(h,\settrain, \setaudit)=1$ (and thus pass the audit), \emph{without} additionally satisfying $F$, which is the actual intended security property (see Fig.~\ref{fig:auditprocessmal} for a visualization). A malicious model provider has strong incentives to do so: for example, \textcolor{BlueViolet}{the insurance company could deploy a model that maximizes accuracy on the audit dataset (and thus passes the audit), yet still unjustifiably denies numerous insurance claims.} 

\noindent\textbf{Attack Game with Known Audit Data}
Before providing a concrete attack example, we introduce a theoretical tool -- an attack game -- which showcases the gap between verifying $f(h,S_{train}, S_{audit})$ (which is what prior approaches certified) and $F(h)=(\Pr\limits_{S_{test}\leftarrow \mathcal{D}}[f(h, S_{train}, S_{test})=1]>\mathsf{p})$ (the intuitive property that one would want to ensure) for CMC schemes where the model provider is given the audit dataset at the beginning of the audit process. 
 
 For simplicity, we will assume that the audit process verifying $f(h,S_{train}, S_{audit})$ is perfectly secure, i.e., the outcome of $\interact{\cP(h,\settrain)}{\cV}(\mathsf{com}_h, \mathsf{com}_S, \setaudit)$, where $\mathsf{com}_S$ is a commitment to $S_{train}$ and $\mathsf{com}_h$ is a commitment to $h$, is 1 if and only if $f(h,S_{train}, S_{audit})=1$.

 In the game, the adversary will win only if it can come up with a model $h$ and training data $S_{train}$, such that: \textbf{(1)} $f(h,S_{train}, S_{audit})=1$, i.e, the adversary would pass an audit on the dataset $S_{audit}$, and \textbf{(2)} $F(h,S_{train})=0$. To make the attack even stronger, we require the adversary to additionally satisfy a utility requirement (formalized via a predicate $L$) in order to win the game. Intuitively, the goal of $L$ is to capture the actual intent of the malicious model owner: For example, in case of the insurance company that wishes to deny claims, we could use $L(h)=\Pr\limits_{x \sim \{0,1\}^d}[h(x) = 0]>0.9$.

\begin{defn}[Adaptive Training with Known Auditing Data]\label{def:forge-audit-new}
Let $\mathcal{M}$ be the model space, $\mathcal{X}$ be the data space and $\mathcal{D}$ be a distribution over $\mathcal{X}$.
Let $f:\mathcal{M}\times \mathcal{X}^k \times \mathcal{X}^\ell \to \{0,1\}$ be a predicate verified by the model certification, and let $F:\mathcal{M}\times \mathcal{X}^k  \to \{0,1\}$ be the actual intended security property. 
Let $L$ denote the utility predicate\footnote{We assume that distribution $\mathcal{D}$ is implicitly ``known'' to $F$ and $L$ (it is either hard-coded or provided as a parameter to each). 
For simplicity of notation, we omit $\mathcal{D}$ from the description of $F$ and $L$.}. Consider the following game played by an adversary $\mathcal{A}$:


\begin{enumerate}
	\item Sample $S_\text{audit} \sim \mathcal{D}^\ell$.
	\item Given $S_\text{audit}$, $\mathcal{A}$ outputs a hypothesis $h_A$ and a training dataset $S_\text{train}$.
	\item Obtain $b = f(h_A,S_{train}, S_{audit})$.
	\item The output of the game is 1 ($\mathcal{A}$ `wins') iff $b=1$, $F(h_A,S_{train})=0$, and
	$L(h_A,S_{train}) = 1$.
    
	The output is $0$ ($\mathcal{A}$ `loses') otherwise. 
\end{enumerate}
\end{defn}

\noindent\textbf{Restricting Knowledge of Audit Data}
Although the attack game in Def.~\ref{def:forge-audit-new} allows an adversary to learn the entire audit dataset $\setaudit$,
some of our results indicate that the attack can be successful even in a more challenging scenario where the adversary only has \emph{partial} knowledge of the audit dataset. 
For example, our experimental results in \S~\ref{sec:eval} exhibit that the attacks on accuracy and fairness certification are effective even when the adversary only has access to a small fraction of the audit dataset.
Moreover, our attack on differential privacy certification in \S~\ref{sec:attack-methods} does not require any knowledge of the audit dataset, as the adversary does not need to forge data based on the audit dataset in order to succeed.

\subsection{Example Data Forging Attacks}
\label{sec:attack-methods}
To demonstrate the breadth of the attack surface for data forging, we now give three concrete examples of data forging attacks: two enable an attacker to evade audits on accuracy (e.g. \cite{ZhangFZS20}) and fairness respectively within the framework of Def.~\ref{def:forge-audit-new}. The third extends beyond this framework (specifically, it does not require the framework's first step, and the adversary \textbf{does not} obtain audit data) to attack differential privacy certification (e.g. \cite{shamsabadi2024dpproof}). Our example attacks target decision tree (DT) models as a proof of concept for ease of exposition and due to the ubiquity of real-world DT deployments. We show empirically that other model types can be effectively attacked as well in Sections~\ref{sec:eval} and~\ref{sec:eval-appendix}.

We first consider our running example of \textcolor{BlueViolet}{an insurance company audit. Say the company uses a \emph{decision tree} model\iffull (see \S\ref{sec:dtbackground} for background).\fi The auditor wishes to check that the model is highly accurate, i.e., $F(h)=\Pr\limits_{x \sim \{0,1\}^d}[h(x) = y_x]>0.95$, where $y_x$ is the true label of $x$. The insurance company saves on costs when it denies claims, i.e., the model provider's utility $L(h)=(\Pr\limits_{x \sim \{0,1\}^d}[h(x) = 0]>0.9)$}.\footnote{For simplicity, we consider datapoints in $\{0,1\}^d$} Following Def.~\ref{def:forge-audit-new}, given audit dataset $S_{audit}$, we wish to construct a hypothesis $h_A$ that passes the audit, while simultaneously achieving $F(h_A,S_{train})=0$ and $L(h_A,S_{train})=1$.

~

\noindent\textbf{Attacking Accuracy Certification} At a high level, Algorithm~\ref{alg:attack} constructs a dataset $S_{train}'$ which causes $h_A$ to overfit to $S_{audit}$, allowing the attacker to achieve distinct levels of performance on the audit data compared to the rest of the distribution. In more detail, we build $S_{train}'$ from the audit dataset $S_{audit}$ using the corresponding true labels—to ensure the model passes the accuracy test on $S_{audit}$—and augment it with points in the neighborhood of those audit points. The points in the neighborhood are labeled $0$, so that outside the audit samples the model learns to label every point with a $0$. More formally, to construct $S_{train}'$, for every point $r\in S_{audit}$ and every (numeric) feature $i$, we add $r+\varepsilon \vec{b}_i$ with label $0$, where $\varepsilon$ is some small number and $\vec{b}_i$ is the $i$th basis vector. We also add every $r \in S_{audit}$ to $S_{train}'$ with honest labels (see Algorithm~\ref{alg:attack} for formal description). Thus, the time complexity of the attack is $\mathcal{O}(|S_{audit}|d)$. After constructing the forged data, we train a tree on $S_{train}'$ via a standard DT training algorithm until leaves are homogeneous.

Consider a baseline CMC protocol for accuracy auditing, wherein a model owner must prove (i) that a decision tree $M$ was trained correctly from a confidential dataset, and (ii) that the predictions $M$ makes on an audit set are accurate. By inputting data from Algorithm~\ref{alg:attack} and a DT trained on this data to such a CMC protocol, an attacker can manipulate the accuracy of the model on the audit set and the rest of the distribution essentially independently. The CMC will not catch this capability, since the attacker `honestly' executes both the verified training and verified inferences in the audit. The attack lies not in the process that the adversary used to obtain the model from the data, but rather in the data itself. The malfeasance is \emph{upstream} of anything this baseline CMC protocol certifies.

We demonstrate this idea empirically in Section~\ref{sec:eval}, and formally below. In particular, Theorem~\ref{thm:attack-isolation} shows that Algorithm~\ref{alg:attack} allows an adversary to win the attack game of Definition~\ref{def:forge-audit-new}.

\begin{algorithm}
\caption{Data Forging to Evade Accuracy Audits}\label{alg:attack}
\textbf{Input}: Audit set $S_{audit}$, dimension $d$, $\varepsilon >0$\\
\textbf{Output}: Training data $S_{train}'$
\begin{algorithmic}
    \Function{AttackA}{$S_{audit},d,\varepsilon,g$}
        \State $S_{train}' \gets S_{audit}$
        \For{$r \in S_{audit}$}
            \For{$i \in [d]$}
                \State $r_0 \gets (r + \varepsilon \vec{b}_i,0)$ \Comment{$\vec{b}_i$ is the one-hot vector in dimension $i$}
                \State $r_1 \gets (r - \varepsilon \vec{b}_i,0)$
                \State $S'_{train} \gets S'_{train} \cup \{r_0 , r_1\}$
            \EndFor
        \EndFor
        \State \Return $S_{train}'$
    \EndFunction
\end{algorithmic}
\end{algorithm}

As we confirm in Fig.~\ref{fig:denial}, this attack achieves good results. We now show that this attack allows the adversary to win in the attack game of Def~\ref{def:forge-audit-new}. To this end, we first state the following theorem:

\begin{thm}\label{thm:attack-isolation}
    Decision tree training with CART until leaves are homogeneous on the output of Algorithm \ref{alg:attack} yields a tree $\mathcal{T}$ such that for every $x\in \mathbb{R}^d$, $\mathcal{T}(x) = 1$ only if $||x-r||_\infty < \varepsilon$ for some $r\in S_{audit}$. 
\end{thm}

At a high level, the proof shows that if two points land in the same leaf, then any point lying between them on one coordinate must also fall in that leaf. Further, as the only non-zero points in $S_{train}'$ are audit points, every non-zero-labeled leaf contains an audit point. For any $x$ at least $\varepsilon$ away from all audit points, if $\mathcal{T}(x)=1$, one can construct a nearby training point with label $0$ that must lie in the same leaf, giving a contradiction as we trained until homogeneity. See \S\ref{sec:thm1proof} for details.\qed

\looseness=-1
Thus, whenever a model provider generates a training dataset using Algorithm~\ref{alg:attack}, an honestly trained decision tree that grows until homogeneity will achieve perfect accuracy on the audit dataset, yet predict zero for all inputs that lie outside an $\varepsilon$-neighborhood of the audit dataset points. Thus, for an appropriate choice of epsilon (which varies depending on the target distribution), the adversary wins in the game specified in Definition~\ref{def:forge-audit-new} with probability one.\footnote{Assuming that a model which almost always outputs 0 is not highly accurate in our scenario.}

~

\noindent\textbf{Attacking Fairness Certification} A similar method of attack can be used to evade fairness audits as well. Algorithm~\ref{alg:data-forge-fairness} demonstrates that the model owner can manipulate the sensitive attributes of a dataset to achieve a desired level of fairness (measured via e.g. demographic parity) on audit data.

\begin{algorithm}
\caption{Data Forging to Evade Fairness Audits}\label{alg:data-forge-fairness}
\textbf{Input}: Audit set $S_{audit}$, dimension $d$, $\varepsilon >0$, sensitive attribute $k$\\
\textbf{Output}: Training data $S_{train}'$
\begin{algorithmic}
    \Function{AttackF}{$S_{audit},d,\varepsilon,g$}
        \State $S_{train}' \gets \emptyset$
        \For{$x,y \in S_{audit}$}
            \State $y' \overset{\$}{\gets} \{0,1\}$
            \State $S_{train}' \gets S_{train}' \cup (x,y')$
        \EndFor
        \For{$r \in S_{audit}$}
            \For{$i \in [d]$}
                \State $r_0 \gets (r + \varepsilon \vec{b}_i,r[k])$
                \State $r_1 \gets (r - \varepsilon \vec{b}_i,r[k])$
                \State $S'_{train} \gets S'_{train} \cup \{r_0 , r_1\}$
            \EndFor
        \EndFor
        \State \Return $S_{train}'$
    \EndFunction
\end{algorithmic}
\end{algorithm}
As in Algorithm \ref{alg:attack}, this attack enforces certain behavior on the audit set while learning a different function on the rest of the distribution. Instead of returning the true labels on the audit dataset (which encourages the audit to appear highly accurate), we return a random label, causing the model to appear completely fair, as the $k=0$ and $k=1$ groups have the same outcome distribution in the audit dataset. Outside of the audit set, however, the model learns to label points according to their sensitive attribute, which results in maximally unfair decisions.

~

\noindent \textbf{Attacking Differential Privacy Certification} We also demonstrate a data forging attack which subverts zero-knowledge proofs of differential privacy certification. The goal of differential privacy is to protect the privacy of users in a dataset, most commonly by adding a controlled amount of random noise to the outputs~\cite{dwork2014algorithmic}. This induces a tradeoff between utility and privacy -- tighter protections on privacy require the addition of greater amounts of noise, which degrades the utility of the output. This tradeoff creates a conflict of interest, since institutions are often incentivized to choose high utility over privacy protections~\cite{greenberg.2017}. 

To address this, Shamsabadi et al. describe a scheme to certify the differential privacy of a model~\cite{shamsabadi2024dpproof}. It works by having the model owner perform a zero-knowledge proof, which shows that they correctly performed differentially private stochastic gradient descent (the standard algorithm for training differentially private models~\cite{song2013stochastic}) on a pool of committed data. 
In our syntax, their protocol be captured by instantiating the predicate $f$ for ``correct training procedure'' (see \S~\ref{sec:predicates}). 
However, this scheme's privacy analysis is broken if the adversary alters a dataset via Algorithm~\ref{alg:data-forge-dp}, resulting in a model with dramatically worse privacy than what is claimed by the certificate. 

The attack is straightforward: the attacker creates several copies of the data points in $S_{train}$ whose privacy they want to erode (denoted as the target set $T$), and then uses the modified dataset as an input to the scheme in~\cite{shamsabadi2024dpproof}. In differential privacy, the \emph{sensitivity} of a training algorithm (i.e. how much model parameters can vary if a single input is changed) is used to calibrate the amount of noise which needs to be added to the model parameters in order to protect the privacy of training data. By copying points in the dataset, Algorithm~\ref{alg:data-forge-dp}breaks the sensitivity assumptions that differentially private training algorithms rely upon. This is similar to the phenomenon noted in Wilson et al.~\cite{wilson2020differentially} who observed that standard differential privacy guarantees over SQL queries fail when multiple queries belong to a single user.

\begin{algorithm}

\caption{Data Forging to Weaken Differential Privacy}\label{alg:data-forge-dp}
\textbf{Input}: dataset $S_{train}$, set of targets $T \subset S_{train}$, duplicity $n$ \\
\textbf{Output}: Training data $S_{train}'$
\begin{algorithmic}
    \Function{AttackDP}{$S_{train},T,n$}
        \State $S_{train}' \gets S_{train}$
        \For{$x \in T$}
            \State add $n-1$ copies of $x$ to $S_{train}'$
        \EndFor
        \State \Return $S_{train}'$
    \EndFunction
\end{algorithmic}
\end{algorithm}

Together these three straightforward example attacks demonstrate the breadth of the data forging attack surface. The model owner's full control over data can be used to attack certificates of multiple properties, with and without prior knowledge of audit data.

\subsection{Statistical Methods Fail to Detect Data Forging}\label{sec:statdetection}

One might seek to strengthen CMC protocols against attacks in the previous section by leveraging statistical tests, e.g., Welch’s $t$-test to assess whether the model owner's dataset was drawn from a different distribution than the audit set. Indeed previous studies propose zero-knowledge proofs of statistical tests for this purpose~\cite{chang2023holmes,duddu2024attesting}. \iffull In \S\ref{sec:detection}\fi however, we prove that adversaries can forge datasets which successfully circumvent CMCs augmented with a variety of statistical tests. 

Our results show that an attacker can forge datasets that successfully evade CMC accuracy audits, while yielding outputs which are indistinguishable from the audit set for a broad family of functions. At a high level, we demonstrate that the attacker can cause the distributional properties of the training data to converge towards the audit data without sacrificing the efficacy of the attack in Algorithm~\ref{alg:attack} by adding more copies of the audit data to the training data. This causes the audit data and the training data to appear as if they were drawn from the same distribution under a variety of statistical tests without impacting the model's ability to learn the desired behavior.

Thus, functions within this family cannot effectively thwart attackers satisfying Definition~\ref{def:forge-audit-new}. We apply this result in Corollary~\ref{cor:welch-main} which shows that a data forging attacker can always evade detection by Welch's t-test.

\begin{cor} \label{cor:welch-main}
    Given an audit dataset $S_{audit}$ and significance level $\alpha$, we can use Algorithm \ref{alg:attack} to construct a training dataset $S_{train}'$ such that for any feature $j$, $S_{train}'$ passes Welch's t-test when its values in feature $j$ are compared to those of $S_{audit}$ with significance level $\alpha$. 
\end{cor}

Proof of Corollary~\ref{cor:welch-main} and the results leading up to it are deferred to \iffull Appendix~\ref{sec:detection}\fi for brevity. 

The difficulty of detecting data forging attacks against CMCs, combined with the variety of certification protocols they can evade, motivate a new perspective on CMC security which we present in the following section.

\section{Secure Cryptographic Model Certification}\label{sec:construction-overview}\label{sec:sec-definition}


In this section, we present our positive results on \emph{provably secure} cryptographic model certification (CMC) protocols, with full details provided in \iffull\S~\ref{sec:construction}. \fi
We begin by introducing the syntax and security definitions for CMC in \S~\ref{sec:securitydefs-overview}.
We then present a template for secure CMC and our main theorem in \S~\ref{sec:construction-theorem-overview}, followed by an example instantiation for accuracy auditing in \S~\ref{sec:auditcsp-accuracy-overview}. 
As we will later see in the case study, our protocol template allows us to restore/expand the security of some of the constructions from Table~\ref{tab:ppml-landscape}.

\subsection{Syntax and Security Definitions}\label{sec:securitydefs-overview}
\begin{defn}[Cryptographic Model Certification Protocol]
A cryptographic model certification (CMC) protocol $\Pi$ for a distributional predicate $F$ is a tuple of algorithms $(\commit, \prove, \audit)$: a commitment, a proving, and an auditing algorithm.
Let $(\com, b)\gets\interact{\prove(h,\settrain)}{\audit}$ denote the interaction between $\prove$ and $\audit$, where $\prove$ takes a hypothesis $h$ and optionally a training dataset $\settrain$ as private input and outputs a commitment $\com$ during an execution, and $\audit$ halts by outputting a binary $b$.
\end{defn}

We propose the following security properties for a CMC protocol $\Pi$ for a distributional predicate $F$. 
We focus on high-level intuition and defer formal cryptographic definitions to \iffull \S~\ref{sec:securitydefs}.\fi
While completeness, binding, and zero knowledge are direct adaptations of those of commitments and zero-knowledge proofs\iffull(see \S~\ref{sec:zksecurity} and \S~\ref{sec:comsecurity})\fi, 
$\tilde{F}$-relaxed knowledge soundness is a novel property tailored to our CMC setting.
In particular, we call this property ``$\tilde{F}$-relaxed'' to allow parameterizing the soundness requirement by a predicate $\tilde{F}$ that is only a close approximation of the desired property $F$.
As we shall see in concrete examples, this relaxation turns out necessary because the auditor only checks the property $f$ on a \emph{finite sample}, which may not perfectly reflect the property $F$ on the \emph{underlying distribution}.
In our concrete instantiations, we will quantify the gap between $F$ and $\tilde{F}$\iffull (see \S~\ref{sec:auditcsp-accuracy} and \S~\ref{sec:auditcsp-demographic-parity}).\fi

\medskip
\noindent\textbf{Completeness.} If an honest prover holds a model $h$ and a training set $\settrain$ such that $F(h,\settrain)=1$, then this prover should pass the audit (i.e, $\audit$ outputs $b=1$) except with small error probability $p$.

\medskip
\noindent\textbf{Binding.} No computationally bounded prover can change its model $h$ or $\settrain$ after generating a commitment $\com$ to them.

\medskip
\noindent\textbf{Zero Knowledge.} An honest execution of the protocol between $\prove$ and $\audit$ does not reveal any information about the model $h$ beyond the fact that $F(h,\settrain)=1$.

\medskip
\noindent\textbf{$\tilde{F}$-relaxed Knowledge Soundness.}
If a (potentially dishonest prover) $\prove^*$ commits to a pair of invalid model $h$ and training dataset $\settrain$, i.e., $\tilde{F}(h,\settrain)\neq1$, then $\audit$ should detect it by outputting $b=0$ with high probability. 
In slightly more detail, for any $\prove^*$, there exists an efficient \emph{knowledge extractor} $\ext$ that extracts a pair $(h,\settrain)$ from $\prove^*$ such that the above condition holds (except with small error probability $\kappa$).

\subsection{Secure CMC Template}\label{sec:construction-theorem-overview}
We now define a \emph{commit-sample-prove} CMC protocol $\auditcsp=(\commit, \prove, \audit)$.
The protocol starts with the model provider committing to its model $h$ using a commitment scheme $\commit$\iffull  (\S~\ref{sec:comsecurity}).\fi
Upon receiving the commitment, the auditor responds with an \emph{audit dataset} $\setaudit$, consisting of $n$ i.i.d. samples drawn from a distribution $\mathcal{D}$ over the query space $\querysp= \{(x_i,y_i)\}_{i=1}^m$.
Finally, the model provider and the auditor run a zero-knowledge proof (ZKP) protocol $\ZKP=(\cP,\cV)$ to prove that the committed model $h$ satisfies an \emph{empirical predicate} $f$ on the sampled audit dataset $\setaudit$, i.e., $\ZKP$ for the relation $\rel$ such that $(\stm,\wit)\in\rel \iff f(h,S_\text{audit})=1\land \com = \commit(h;\rho)$.
Here, the empirical predicate $f$ should be chosen to closely approximate the desired property $F$ over the underlying distribution $\mathcal{D}$; we will formalize this requirement as small \emph{false negative rate}.
Moreover, $f$ should be chosen such that the \emph{false positive rate} of $f$ is small with respect to the relaxed predicate $\tilde{F}$. 
Intuitively, this means that if a model $h$ does not satisfy the relaxed property $\tilde{F}$, then it should fail the audit with high probability (i.e., soundness), while if a model $h$ satisfies the desired property $F$, then it should pass the audit with high probability (i.e., completeness).
While we focus on a protocol checking $F$ on a hypothesis $h$ only, the construction below can be naturally extended to a more complex $F$ and $f$ that additionally take a training dataset $\settrain$ as input.

\medskip
\noindent\underline{$\interact{\prove(h)}{\audit}$}
\begin{enumerate}
    \item $\prove$ computes $\com=\commit(h;\rho)$ using a uniformly random string $\rho$ and sends $\com$ to $\audit$.
    \item $\audit$ samples $\setaudit \sim \mathcal{D}^n$ and sends it to $\prove$.
    \item $\prove$ and $\audit$ execute $b\gets\interact{\cP(\wit)}{\cV}(\stm)$, where $\stm = (\com,S_\text{audit})$ and $\wit = (h,\rho)$. Here, $\prove$ plays $\cP$ and $\audit$ plays $\cV$ of $\ZKP$. 
    \item $\prove$ outputs $\com$, while $\audit$ outputs $b$.
\end{enumerate}

\newcommand{\pfnr}{p_{\text{fnr}}}
\newcommand{\pfpr}{p_{\text{fpr}}}

The key takeaway is that the audit dataset $\setaudit$ should be chosen \emph{independently} of the model $h$ and the training dataset $\settrain$.
Sampling $\setaudit$ after $h$ is committed ensures independence, rendering our earlier attack ineffective. The following theorem (formally stated and proved in \iffull\S~\ref{sec:constructiondetails}\fi) states that, if the empirical predicate $f$ satisfies the bounds mentioned below, $\auditcsp$ satisfies the desired security properties in a general fashion.
\textbf{Essentially, our result allows the protocol designer to choose an empirical predicate $f$ that approximates the desired property $F$ well enough, and then plug in any commitment scheme \iffull(\S~\ref{sec:comsecurity})\fi and ZK proof system \iffull(\S~\ref{sec:zksecurity})\fi satisfying the required security properties \iffull(\S~\ref{sec:zksecurity},\S~\ref{sec:comsecurity})\fi to get a secure auditing scheme.}

\iffull
\noindent\textbf{Theorem \ref{thm:auditcsp-general} (informal).} \textit{
Let the empirical predicate $f$, the distributional predicate $F$, and the relaxed distributional predicate $\tilde{F}$ satisfy the following false negative and false positive rate bounds for every model $h$:
\begin{align*}
	& \Pr_{\setaudit \sim \mathcal{D}^n}[f(h,\setaudit)\neq 1 \mid F(h)=1] \leq \pfnr \quad\\
	& \Pr_{\setaudit \sim \mathcal{D}^n}[f(h,\setaudit)=1 \mid \tilde{F}(h) \neq 1] \leq \pfpr
\end{align*}
If $\commit$ and $\ZKP$ satisfy the standard security properties (\S~\ref{sec:zksecurity}-\ref{sec:comsecurity}), then the protocol $\auditcsp$ is \textbf{complete} with error probability $\pfnr$, \textbf{binding}, \textbf{zero knowledge}, and \textbf{$\tilde{F}$-relaxed knowledge sound} with error probability $\approx\pfpr$.
}
\fi

\smallskip
\noindent\textit{Proof sketch.}
We focus on $\tilde{F}$-relaxed knowledge soundness, as the other properties directly follow from the underlying commitment scheme and ZK proof system with standard security guarantees.
Our proof carefully unifies standard arguments from learning theory with the proof technique of ``knowledge extraction with interleaved probabilistic tests'', commonly used in advanced proof systems~\cite{Nguyen,AlmostSpecialSoundness,AdaptiveSpecialSoundness}.
To illustrate the intuition, let us assume the commitment scheme is perfectly binding and straightline extractable, i.e., once the prover sends $\com$ to the auditor, one can immediately extract a unique model $h$ and a unique randomness $\rho$ such that $\com = \commit(h;\rho)$.
Such a commitment scheme can be constructed in the common reference string (CRS) model using public key encryption.
If the committed model $h$ does \emph{not} satisfy $\tilde{F}$, then by the assumption on the false positive rate, the probability that $f(h,S_\text{audit})=1$ is at most $\pfpr$ over the choice of $\setaudit$.

The actual proof is significantly more involved, when the commitment scheme is \emph{only computationally binding} (which is the case for most practical instantiations). 
In particular, we need to construct a meta-extractor that runs the knowledge extractor $\cE$ for $\ZKP$ to obtain a candidate model $h$, and rewinds $\cE$ with fresh $\setaudit'\sim\mathcal{D}^n$ to ensure the validity of $h$ via probabilistic tests (or otherwise break binding of the commitment).
Our formal security proof takes care of these subtle technicalities.


\subsection{Example: CMC for Accuracy}\label{sec:auditcsp-accuracy-overview}
We detail example instantiations of $f$ and $F$ for accuracy \iffull(\S~\ref{sec:auditcsp-accuracy})\fi and demographic parity auditing\iffull (\S~\ref{sec:auditcsp-demographic-parity})\fi, both of which enable arbitrary small false positive and negative rates with a sufficient number of samples $n$.
To instantiate $\auditcsp$ for accuracy auditing, we consider the empirical and true error rates as follows:
\begin{align*}
	\hat{\ell}_S(h) = \frac{1}{n}\sum_{(x,y)\in S}\mathbb{I}(h(x)\neq y)\quad
	\ell(h) = \mathbb{E}_{(x,y)\sim \mathcal{D}}[\mathbb{I}(h(x)\neq y)] 
\end{align*}
where $n = |S|$, and define the empirical predicate $f$, the distributional predicate $F$, and the relaxed predicate $\tilde{F}$ as follows.
\begin{align*}
    f(h,S_\text{audit}) = 1  & \iff \hat{\ell}_S(h) \leq t+\delta \\
    F(h) = 1 & \iff \ell(h) \leq t  \\
	\tilde{F}(h) = 1 & \iff \ell(h) \leq t + 2\delta
\end{align*}
With these definitions in place, we can bound both the false negative $\pfnr$ and false positive rates $\pfpr$ by $2 e^{-2n\delta^2}$ by a standard concentration bound argument \iffull(see Lemma~\ref{lem:fn-fp-accuracy} in \S~\ref{sec:auditcsp-accuracy}).\fi
By plugging them into Theorem~\ref{thm:auditcsp-general}, we conclude that $\auditcsp$ instantiated for accuracy auditing is a secure CMC with the following properties for a sufficiently large $n\delta^2$. 
\begin{itemize}
    \item If the true error of the model $h$ is $\leq t$, then the auditor accepts with high probability.
    \item If the auditor accepts, then the auditor gets assurance that the true error of the model $h$ is at most $t + 2\delta$, even if the prover is misbehaving.
\end{itemize}
As usual in learning theory, there is a trade-off between the sample size $n$ and the tightness of the confidence parameter $\delta$. 
By tuning $\delta$ appropriately, one can reduce the number of samples $n$ needed to achieve desired false positive and negative rates.
This is particularly important in practice since 1) the computational cost of $\ZKP$ typically scales with $n$, and 2) gathering large audit datasets can be challenging.
For instance, by choosing $(n,\delta)=(10\text{k},0.02)$ or $(n,\delta)=(40\text{k},0.01)$, we can achieve both false positive and negative rates $\leq 0.001$.
Since typical dataset sizes in practice are in the order of tens or hundreds of thousands, these configurations enable high-confidence auditing while maintaining reasonable computational costs for $\ZKP$.

\section{Case Studies: Vulnerability to Data Forging in Previous Work}\label{sec:casestudy}
Our formalization from \S\ref{sec:sec-definition} and \S\ref{sec:attack} lets us test whether a protocol is vulnerable to data-forging, and provides a way to lift the security guarantees of prior works to resist data forging attacks. For CMC works which reveal neither the model nor the training data, the most important check boils down to whether the prover is required to commit to the training data and/or to the model before seeing the audit dataset. 

We examined several prior works~\cite{ZhangFZS20,shamsabadi2022profitt,yadav2024fairproof,LiuXZ21,franzese2024oath,shamsabadi2024dpproof,kangdnn,WangH23,p2nia} with formal security guarantees. Surprisingly, the majority of them either do not explicitly state when the audit dataset is revealed, or assume the prover's training dataset and/or the model itself to be trusted (and are susceptible to data forging if the prover is actually malicious). Works that do not discuss the timing of the commitment often point out that their solution can be used to conduct audits using publicly known datasets, in which case the public dataset can be assumed to be known to the adversary prior to auditing, making the solution vulnerable to data-forging. 

We present two of our case studies below, and discuss the rest in \S\ref{sec:case_study}. We note that works that are not susceptible to our data-forging attacks nonetheless provide only dataset-specific guarantees, i.e., their proofs certify properties solely on the chosen audit set/inference queries already submitted by clients, without extending to the underlying data distribution. 
It would be interesting to perform an analysis similar to that in \iffull \S\ref{sec:auditcsp-accuracy} and \S~\ref{sec:auditcsp-demographic-parity} \fi to derive formal guarantees that hold for the corresponding distributions.
 We summarize the results of our findings in Table~\ref{tab:ppml-landscape}.

\subsection{Zero Knowledge Proofs for Decision Tree Predictions and Accuracy}

\noindent\textbf{Goal and Solution Details.} Zhang et al.~\cite{ZhangFZS20} introduce protocols for auditing accuracy and verifying decision tree predictions without revealing any additional information about the model itself. Zhang et al.'s main contribution is in designing a custom zero-knowledge proof tailored to efficiently verifying the decision tree prediction. The proof consists of algorithms to generate public parameters, custom commitment algorithm for decision tree models, the prover's algorithm which outputs a proof of inference/accuracy, and verifier algorithm to check this proof. The prover, i.e., model provider, must first commit to its model and subsequently demonstrate that the predictions on client queries are consistent with this commitment. For accuracy verification, the authors propose a batching technique that allows to more efficiently checks the correctness of model predictions across multiple inputs, followed by a ZKP-verified calculation of the proportion of predictions that match the ground truth labels of the data.

\noindent\textbf{Security Model.} Zhang et al.'s security notion follows the traditional definition structure of a zero-knowledge proof system, where the protocol is required to satisfy correctness, soundness, and zero-knowledge. Both parties can be malicious. In the context of our analysis we are interested in soundness, which specifies whether a malicious prover can deceive the verifier (i.e., auditor), that the prover's hypothesis passed the test. At a high level, the authors' soundness definition can be summarized as follows: A prover should not be able to output a commitment to a tree $\mathcal{T}$ along with a proof $\pi$, an accuracy threshold $acc$, and a dataset $D$, such that the verifier accepts the proof, and at the same time, $\mathcal{T}'s$ accuracy on $D$ is not equal to $acc$. 

\noindent\textbf{Discussion.} The security notion that Zhang et al's work provides aligns well with the intuitive goals of verifying both the correctness of individual predictions and the accuracy of a model on a given dataset. However, it does not provide formal guarantees for datasets beyond the audited one, i.e., the accuracy verification solution does not generalize to other datasets drawn from the same distribution. The protocol description further does not explain how the audit dataset is sampled.  In fact, Zhang et al. explicitly note that it is possible to use their solution to check accuracy on a \emph{public dataset}. In this setting, their approach falls within our framework of Definition~\ref{def:forge-audit-new}, and is vulnerable to the same attack as outlined in \S\ref{sec:attack-methods}. In fact, note that our example works even if the prover supplies an \emph{additional} proof of training to complement its proof of accuracy. 

\noindent\textbf{Expanding/Restoring Security.} Using our techniques, Zhang et al.'s construction can be easily adapted to satisfy the guarantee with respect to new datasets drawn from the same distribution. In fact, this corresponds precisely to our example in \S~\ref{sec:auditcsp-accuracy-overview}, where Zhang et al.'s construction can be plugged in as the ZKP of our protocol $\auditcsp$ in \S\ref{sec:construction-theorem-overview}.

\subsection{Confidential-DPproof: Confidential Proof Of Differentially Private Training}
\noindent\textbf{Goal and Solution Details.} Shamsabadi et al.~\cite{shamsabadi2024dpproof} present Confidential-DPproof, 
a framework that enables the model provider to prove to an auditor that their model was correctly trained via DP-SGD~\cite{abadi2016deep}, a classic approach for training models with differential privacy guarantees. The certification of DP-SGD's training run is done in zero-knowledge.

\noindent\textbf{Security Model.} Shamsabadi et al.~\cite{shamsabadi2024dpproof} consider a setting with a mutually distrusting model provider and auditor. The model provider is fully malicious, while the auditor is semi-honest. Confidential-DPproof considers standard definitions of correctness, soundness, and zero-knowledge.

\noindent\textbf{Discussion.}
As we describe in \S\ref{alg:data-forge-dp} (``Attacking Differential Privacy Certification''), Shamsabadi et al.'s solution is susceptible to data forging attacks. 

\noindent\textbf{Limitations of Our Secure Template in the Privacy Setting.} Defining practical empirical predicates for privacy testing is more challenging for privacy than for accuracy or fairness. Empirical audits of differential privacy, such as membership inference attack success rate~\cite{hu2022membership}, often incur hefty computational costs. Further, both privacy audits and differentially private training necessitate assumptions about training data which are difficult to validate in the CMC model without further cryptographic infrastructure. 

As a result, the approach we introduced above for Zhang et al.~\cite{ZhangFZS20} does not apply to Confidential-DPproof. The core of the issue that Confidential-DPproof is essentially a proof of training, and the data for this training is necessarily chosen by the model provider, not by the auditor. One possible way to still make use of Confidential-DPproof guarantees would be to verify that the training data was generated by a set of trusted parties independent from the model provider. For instance, each individual in the dataset could sign their data record with a distinct key pair. The service provider could then prove that each training record verifies under a single individual’s public key. This would thwart the attack in Algorithm~\ref{alg:data-forge-dp}, as duplication of data points would not be possible. However, instantiating such a scheme efficiently and proving its security to more general data forging attacks is a challenge beyond the scope of the present paper.  We reserve it as an interesting direction for future work.



\section{Evaluation}\label{sec:eval}

In this section we underscore the importance of data forging attacks by mounting proof of concept attacks for models trained on a variety of datasets. We show that our attack is effective in making inaccurate models appear accurate and unfair models appear fair, and empirically demonstrate a variety of other qualities, e.g. undetectability with a variety of statistical tools.
\iffull
We provide source code for our implementation and for reproducing our experiments.\footnote{\url{https://zenodo.org/records/20337318}} 
\fi

\begin{figure*}[t]
    \centering
    \includegraphics[width=\linewidth]{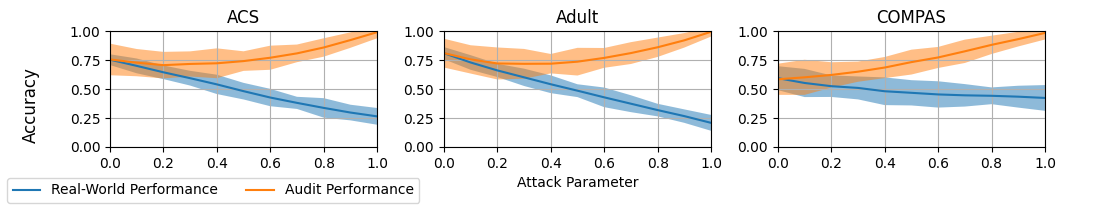}
    \caption{Accuracy of models trained on datasets constructed to minimize real-world accuracy while maintaining high accuracy on audit data for several benchmarks. Attack parameter controls what portion of audit data is known to the model owner and what portion of training data is labeled maliciously. Values are averages over ten runs, error bars represent one standard deviation. 
}
    \label{fig:accuracy}
\end{figure*}

\begin{figure*}[t]
    \centering
    \includegraphics[width=\linewidth]{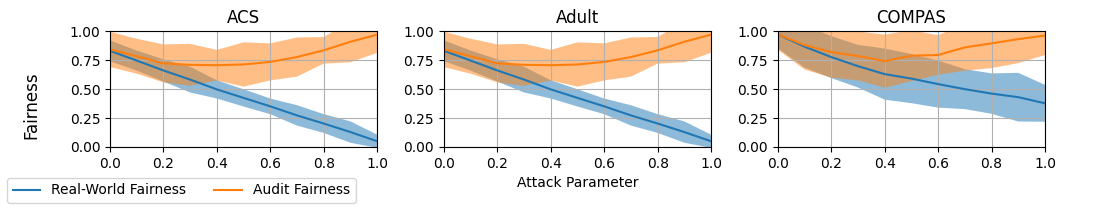}
    \caption{Fairness of models trained on datasets constructed to maintain demographic parity on audit data while producing unfair results on real-world data using various benchmarks. Attack parameter controls what portion of audit data is known to the model owner and what portion of training data is labeled maliciously. Values are averages over ten runs, error bars represent one standard deviation. Fairness means $1-\text{demographic parity gap}$.
}
    \label{fig:fairness}
\end{figure*}

\myparagraph{Experimental Setup.} We use six well-known fairness benchmarking datasets in our experiments: ACSEmployment~\cite{ding2021retiring}, Adult~\cite{Adult_income}, COMPAS~\cite{angwin2016machine}, German Credit~\cite{german}, Default Credit~\cite{Default_credit}, and Communities \& Crime~\cite{communities_and_crime}. We implemented a modified version of our attack from \S\ref{sec:attack-methods} in Python 3.12.3 using SciKit-Learn version 1.6.1 that attempts to minimize an objective when deployed and evaluated its performance against a variety of datasets. For a given run, we split the dataset into an evaluation dataset consisting of $30\%$ of the data, an audit dataset containing 1000 data points (500 for German Credit due to its smaller size), and an initial training data set. We represent the interpolation between a fully honest training run and a fully malicious one by the \emph{attack parameter} $p \in [0,1]$.
Concretely, $p$ controls two components of the forged training set $S'_{\mathit{train}}$: (i)~the fraction of audit points that contribute $\pm\varepsilon$-perturbed neighbors (with flipped labels), and (ii)~the fraction of $S_{\mathit{train}}$ rows whose labels are flipped.
\iffull At $p=0$ the construction degenerates to honest training (no perturbed neighbors, no label flips), i.e., $S'_{\mathit{train}} = S_{\mathit{train}}$; at $p=1$ $S_{\mathit{audit}}$ gets included in $S'_{\mathit{train}}$, every audit point generates $2d$ flipped-label neighbors, and every $S_{\mathit{train}}$ label is inverted.\fi
The specifics of how honest and malicious data points are labeled depends on the objective.

To attack accuracy, we constructed a training dataset using a modified Algorithm~\ref{alg:attack}, labeling the additional data from the attack with $1-r_y$ rather than $0$ and adding them to the initial training data set. For our fairness attack, we constructed the training data similarly, changing how honest and malicious data points were labeled to match Algorithm~\ref{alg:data-forge-fairness}. Honest data points were given random labels, while malicious data points were labeled according to their sensitive attribute. Both attacks then fit a decision tree to their constructed training data using SciKit-Learn's decision tree classifier class.

To ensure that our attack would evade statistical detection, we added extra copies of the audit data to the training data, as computed in \iffull Corollary~\ref{cor:welch-t-test} \fi to pass Welch's $t$-test with significance level $0.05$.

\myparagraph{Attacking Accuracy Audits.} We ran our attack on six benchmark datasets -- three shown in Figure~\ref{fig:accuracy} (remaining in Figure~\ref{fig:accuracy-attack-additional}). Across all datasets, our data forging attack enforces high audit accuracy while simultaneously encouraging low performance on real-world evaluation data. Thus {\bf our attack successfully makes inaccurate models \emph{appear} accurate to an auditor.} In practice, this could mean that a fraudulent insurance provider could enrich themselves by using a highly inaccurate model to deny claims, despite passing accuracy audits. 

\myparagraph{Attacking Fairness Audits.} We also performed the attack while targeting demographic parity (using sex as the sensitive attribute) on three datasets, which we present in Figure~\ref{fig:fairness}. We were able to reliably train a model with close to $0$ fairness gap on the audit dataset, but close to $1$ fairness gap when deployed. In other words, {\bf our attack successfully makes unfair models \emph{appear} fair to an auditor.} 

\myparagraph{Evading Detection via Statistical Methods.}  \iffull We show how our attack can be executed in ways that evade detection by a variety of statistical approaches in Appendix Table~\ref{tab:stats}. We were able to construct malicious training datasets with summary statistics that match those of the audit dataset very closely, and Welch's $t$-test and Levene's test regularly concluded that the audit and test datasets were drawn from the same distribution. This is consistent with our theoretical results in  Appendix~\ref{sec:detection}.\fi 


\myparagraph{Additional Results.} An adversary can use data forging attacks to achieve concrete goals beyond degradation of accuracy or fairness, as we show in the Appendix~\ref{sec:eval-appendix}. For example, Figure~\ref{fig:denial} shows how an insurance provider could use our attack to hide the claim denial rate of a model from auditors. We have also conducted preliminary experiments that suggest that variations of our attack generalize to XGBoost (\S~\ref{sec:xgboost}) and neural network models (\S~\ref{sec:neural}).

\section{Conclusion and Future Work}
This work brings attention to data-dependent vulnerabilities in CMC. 
We showed that, in several proposed certification frameworks, a provider can pass an audit while undermining its intended guarantees in real-world deployment. We proposed concrete attack strategies and provided preliminary evidence that standard statistical tests are unlikely to reliably detect them. Empirically, we demonstrated the effectiveness of these attacks for certifications for accuracy and fairness.
We then introduced new formal security definitions which allow property certification on data \emph{distributions} (rather than data \emph{sets}), thereby ruling out the data-dependent attacks identified in this work. We then presented a general-purpose protocol satisfying our definition. 

Our results should not be interpreted as showing that prior CMC constructions are technically incorrect. Rather, they show that the guarantees they provide may not align with the guarantees required in realistic auditing scenarios. More broadly, our work illustrates the importance of close communication between the cryptography, security, and machine learning communities when designing future certification frameworks, so that formal guarantees are aligned with the threats that arise in deployment.

The attack strategy presented in this work poses several open questions. 
While we demonstrate the data forging attack is undetected even in the presence of Welch's $t$-test and Levene's test, it remains to be seen whether other statistical tests could effectively detect the attack. 
Based on the results that we have derived, we find it unlikely that other statistical tests will be effective in detecting the attack. However, we reserve such analysis for future work.
We provide rigorous formal proofs that our attacks are effective on decision trees, and preliminary evidence that a similar approach generalizes to neural networks. Characterizing a formal relationship between neural network model capacity and attack effectiveness could be a promising direction in future work.

Our secure CMC template underscores  the importance of keeping audit data hidden until the service provider's model is committed. 
This imposes a limitation on auditing in practice: auditors must either regularly gather fresh data (since the audit dataset is typically revealed during the audit), use additional cryptographic techniques such as secure multiparty computation to keep data hidden during the audit, or perform continuous auditing on user data. Each of these options has strengths and drawbacks which should be evaluated in more detail by future work.
Moreover, while we focused on a CMC template based on fresh audit datasets, it remains unclear how to securely certify distributional properties that do not involve an audit dataset, such as differentially private training. 
A promising direction for future work is a countermeasure against the third attack in \S\ref{sec:attack-methods}, which likely requires a more advanced CMC mechanism to certify the well-formedness of the training data.

\ifanonymous
\else
\section*{Acknowledgments \& Disclaimer}
We thank \ifneurips Nicolas Papernot for the contributions to this work, \fi Chen-Da Liu-Zhang for many helpful discussions, and Jonas Guan, Sierra Wyllie and Mohammad Yaghini for helpful feedback on an earlier draft.
Resources used in preparing this research were provided, in part, by the Province of Ontario, the Government of Canada through CIFAR, and companies sponsoring the Vector Institute. The work was further supported by the CBI postdoctoral fellowship.

This paper was prepared in part for information purposes by the Artificial Intelligence Research group of JPMorgan Chase \& Co and its affiliates (“JP Morgan”), and is not a product of the Research Department of JP Morgan. JP Morgan makes no representation and warranty whatsoever and disclaims all liability, for the completeness, accuracy or reliability of the information contained herein. This document is not intended as investment research or investment advice, or a recommendation, offer or solicitation for the purchase or sale of any security, financial instrument, financial product or service, or to be used in any way for evaluating the merits of participating in any transaction, and shall not constitute a solicitation under any jurisdiction or to any person, if such solicitation under such jurisdiction or to such person would be unlawful.
\fi


\bibliographystyle{plainurl}
\bibliography{ref}

\ificml
		\bibliographystyle{icml2026}
		\appendix
		\onecolumn 
\else
\fi

\appendix
\section{Related Work}\label{sec:relatedworkcontinued}

We now briefly discuss further prior works in the field of model certification.
A number of recent works aim to prove desirable model properties. In terms of \emph{what} these works prove, they can be roughly categorized into proofs of training, inference, accuracy, and fairness. In terms of \emph{how} the corresponding protocols work,
recent works on certifiable ML can be categorized as follows:  

\noindent\textbf{Cryptographic approaches.} A prolific line of research adapts 
various cryptographic techniques to certify properties such as accuracy, fairness, etc., without revealing the model's details. The most common technique is \emph{zero-knowledge proofs} (zk proofs), which allow to formally prove that a model satisfies certain properties without revealing anything else about the model.
They have been used to certify
fairness~\cite{shamsabadi2022profitt,yadav2024fairproof,franzese2024oath,zhang2025scalable}, inference~\cite{ZhangFZS20}, accuracy~\cite{ZhangFZS20}, and to prove that the model has been trained using a certain algorithm \cite{abbaszadeh2024DNNzkpot,garg2023experimenting,sun2024zkdl,pappas2024sparrow} (without revealing the training data). 
Works such as Chang et al.~\cite{chang2023holmes} and Duddu et al.~\cite{duddu2024attesting} use \emph{secure multi-party computation} (MPC)~\cite{Yao82b,GoldreichMW87}, a foundational cryptographic primitive introduced and refined through decades of research~\cite{GargMPP16,GoyalMPS21,ChaumCD87,Ben-OrGW88,EscuderoMP25,Yao86,IshaiPS08,KreuterSS12,GentryHKMNRY21,ChoudhuriGGJK21,BandarupalliJKL25,DamgardN07}. MPC allows mutually distrusting parties to jointly compute on private inputs without revealing anything about the inputs apart from the outcome. 

\noindent\textbf{Black box auditing/Statistical testing.} These approaches probe a model by submitting inputs, collecting outputs, and analyzing them for undesirable behavior. Works by Saleiro et al.~\cite{saleiro2018aequitas} and Tramer et al.~\cite{tramer2017fairtest} use black-box testing to check for potential unfairness or bias, while Tan et al.~\cite{TanCHL18} distill a new model to gain insight into the black box one.

\noindent\textbf{Outside-the-box auditing.} Here the model owner provides access to information beyond query responses, such as source code, documentation~\cite{mitchell2019model}, hyperparameters, training data, deployment details, or internal evaluation results.

\iffull
\section{Additional Preliminaries}\label{sec:more_preliminaries}
\subsection{Welch's $t$-test}
\noindent\textbf{Welch's $t$-test}
The goal of $t$-test is to determine whether the unknown population means of two groups are equal or not.
That is, for random variables $X$ and $Y$, it compares the following hypotheses on their means $\mu_X = \mathrm{E}[X]$  and $\mu_Y = \mathrm{E}[Y]$: 
\begin{itemize}
\item Null Hypothesis $H_0$: $\mu_X =\mu_Y$
\item Alternative Hypothesis $H_1$: $\mu_X\neq \mu_Y$
\end{itemize}

Assuming that $X$ and $Y$ independently follow Gaussian distributions with unknown variances, Welch's $t$-test proceeds as in Algorithm~\ref{alg:t-test}.
\begin{algorithm}[htb]
	\caption{Welch's $t$-test}
	\label{alg:t-test}
	\textbf{Input:}  $\mathcal{X} = \{x_i\}_{i\in[n]}$, $\mathcal{Y}= \{y_i\}_{i\in[m]}$, where $x_i\sim X$ and $y_i\sim Y$, and a significance level $\alpha$
	\textbf{Output}: Null hypothesis $H_0$ (i.e., $\mu_X =\mu_Y$) or alternative hypothesis $H_1$ (i.e., $\mu_X\neq \mu_Y$)
	\begin{algorithmic}[1]
		\State Compute sampled means $\bar{x} = \frac{\sum_i x_i}{n}$ and $\bar{y}=\frac{\sum_i y_i}{m}$
		\State Compute sampled variances $v_x = \frac{\sum_i (\bar{x}-x_i)^2}{n-1}$ and $v_y=\frac{\sum_i (\bar{y}-y_i)^2}{m-1}$.
		\State Compute the test statistic $t = \frac{\bar{x}-\bar{y}}{\sqrt{v_x/n + v_y/m}}$
		\State Compute the degree of freedom $d = \frac{(g_x+g_y)^2}{g_x^2/(n-1) + g^2_y/(m-1)}$, where $g_x = v_x/n$ and $g_y = v_y/m$
		\State Obtain the critical value $t_{\textit{cr}}$ from the $t$-table, given $d$ and $\alpha$.
		\State \textbf{If} {$|t|< t_{\textit{cr}}$} \textbf{return} $H_0$ \textbf{else} \textbf{return} $H_1$ 
	\end{algorithmic}
\end{algorithm}

\subsection{Decision Trees}\label{sec:dtbackground}
In our attack constructions we focus on decision tree models. Decision tree-based solutions are among the most popular machine learning algorithms, particularly known for their effectiveness in classification problems such as loan approval and fraud detection. A decision tree is trained by recursively partitioning the dataset from the root to the leaves. At each step, a split is determined by a splitting rule that aims to maximize an objective function, such as information gain. For prediction, the input follows a path from the root to a leaf, where at each internal node, the decision depends on whether the input satisfies the corresponding threshold (see Algorithm~\ref{alg:dtinference}). 

For completeness, in Algorithm~\ref{alg:dtinference} we present the algorithm for decision tree inference.
		
\begin{algorithm}[htb]
	\caption{Decision Tree Inference}
	\label{alg:dtinference}
    
    \textbf{Input:} Decision tree $h$, input $\ba$. 

    \textbf{Output}: Classification result.
  
    \begin{algorithmic}[1]
        \State Let $\cur: = h.\troot$ \Comment{Set $\cur$ to be root of the tree}
        \While{$\cur$ is not a leaf}
        \If{$\ba[\cur.\attr] < \cur.\thr$}
        \State $\cur:= \cur.\tleft$. \Comment{Set $\cur$ to be current node's left child}
        \Else
        \State $\cur:= \cur.\tright$. \Comment{Set $\cur$ to be current node's right child}
        \EndIf
        \EndWhile
        \State \textbf{return} $\cur.\class$
    \end{algorithmic}
\end{algorithm}

\subsection{Security Properties of Zero-Knowledge Proofs}\label{sec:zksecurity}
Let $\ZKP = (\cP, \cV)$ be an interactive proof system for a relation $\rel = \bigcup_{\secp\in \mathbb{N}} \rel_\secp$.
In what follows, we denote by PPT \emph{probabilistic polynomial time}.

\medskip
\noindent\textbf{Completeness}
$\ZKP$ is (perfectly) \emph{complete} if for any $(\stm,\wit)$ satisfying $\rel$, it holds that: 
$$\Pr[1\gets\interact{\cP(\wit)}{\cV}(\stm)] =1.$$

\medskip
\noindent\textbf{Knowledge Soundness}
$\ZKP$ is (adaptively) \emph{knowledge sound} with knowledge error $\kappa$ if for any (stateful) PPT adversary $\cP^*=(\cP_0,\cP_1)$, there exists an expected polynomial time extractor $\cE$ such that the following holds:
\begin{align*}
	\pext \geq \pacc - \kappa
\end{align*}
where 
\begin{align*}
	\pext &= \probrv{R_\secp(\stm,\wit)=1}{\stm\gets\cP_0(1^\secp); \wit\gets\cE_{\cP}(\stm)}\\
	\pacc &= \probrv{b=1}{\stm\gets\cP_0(1^\secp); b\gets\interact{\cP_1}{\cV}(\stm)}
\end{align*}
where $\cE$ has non-black-box access to $\cP^*$.
Informally, this means that any cheating prover must know a valid witness if it convinces verifier.

\medskip
\noindent\textbf{Zero-Knowledge}
Let $\view_{\cV}^{\cP(\wit)}(\stm)$ be a string consisting of all the incoming messages that $\cV$ receives from $\cP$ during the interaction $\interact{\cP(\wit)}{\cV}(\stm)$, and $\cV$'s random coins.
$\Pi$ is (honest verifier) \emph{zero-knowledge} if there exists a PPT simulator $\cS$ such that for any adversary $\cA$ and any $(\stm,\wit)\in\rel_\secp$, the following is negligible in $\secp$.

\begin{align*}
	\left|
	\begin{aligned}
		&\probrv{b=1}{
			& b\gets\cA(\view_{\cV}^{\cP(\wit)}(\stm))} \\
		- &\probrv{b=1}{
			& \view' \gets\cS(\stm);
			& b\gets\cA(\view')}
	\end{aligned}\right|
\end{align*}

Informally, this means that the protocol execution reveals no information to $\cV$ about $\wit$.

\subsection{Security Properties of Commitment Schemes}\label{sec:comsecurity}
Let $\commit$ be a commitment scheme. For simplicity, we omit the key generation algorithm $\gen$ and present a class of the simplest commitments whose openings are checked by re-computing and comparing (e.g., hash commitment $H(m||\rho)$). More generally, some commitment schemes require a separate verification algorithm $\mathsf{Verify}$ to check the validity of a commitment given some \emph{decommitment} information. Our auditing framework can be extended to such schemes by having the model provider prove the knowledge of the decommitment information in zero knowledge.

\begin{defn}[Commitment Scheme]
A commitment scheme is an algorithm $\commit$,
which is executed as $\com\gets\commit(m;\rho)$. 
It takes as input a message $m\in\{0,1\}^{\ell_m(\secp)}$, a uniformly sampled randomness $\rho \in \{0,1\}^{\ell_r(\secp)}$, and returns a commitment $\com \in \{ 0,1\}^{\ell_c(\secp)}$. Here $\ell_m, \ell_r, \ell_c$ are some polynomials in $\secp$, the security parameter (determining the desired level of security).
\end{defn}

\medskip
\noindent\textbf{Binding}
$\commit$ is \emph{computationally binding} if for any PPT adversary $\cA$, the following is negligible in $\secp$:
\begin{align*}
    \probrv{\com &= \commit(m;\rho) \\ 
    \com &= \commit(m';\rho') \\  m&\neq m'}{(\com,m,m',\rho,\rho')\gets\cA(1^\secp)}
\end{align*}

\medskip
\noindent\textbf{Hiding}
$\commit$ is \emph{computationally hiding} if for any PPT adversary $\cA$, the following is negligible in $\secp$:
\begin{align*}
    \left|
    \probrv{b=b'}{ (m_0,m_1)\gets\cA(1^\secp); \\ \rho\gets\{0,1\}^{\ell_r(\secp)}; b\gets\{0,1\}; \\ \com\gets\commit(m_b;\rho); b'\gets\cA(\com)} - \frac{1}{2}
    \right|
\end{align*}
\fi
\section{Proof of Theorem \ref{thm:attack-isolation}}\label{sec:appendix}\label{sec:thm1proof}

\begin{proof}
    First, let us show that if points $a$ and $b$, where $a_{i} < c_{i} < b_{i}$ and $c_{j}=b_{j}$ for all $j\neq i$, get sorted into the same decision tree leaf, then $c$ is also sorted into that same leaf. Notice that in order for two points $x$ and $y$ to get sorted into different leaves, there must be a node which splits on a feature $i$ such that $x_{i}\neq y_{i}$ and $x_{i}\leq t \leq y_{i}$ or $y_{i} \leq t \leq x_{i}$ where $t$ is the threshold to split upon. Then if $b$ and $c$ were sorted into separate leaves, there must be a node on the path that $b$ takes through the tree that splits on feature $i$ with a threshold $t$ that satisfies $c_{i} \leq t \leq b_{i}$. However, such a node would also sort $a$ distinctly from $b$, so such a node cannot occur. Thus, $b$ and $c$ must be sorted into the same leaf.
    
    Now, note that as we train the tree until its leafs are homogeneous, every datapoint in $S_{train}'$ must be classified correctly (according to the label we assigned to it in Algorithm~\ref{alg:attack}). Further, since the only datapoints in $S_{train}'$ with non-zero labels are datapoints from $S_{audit}$, for every leaf in $\mathcal{T}$ that is associated with a non-zero class, we have at least one $r \in S_{audit}$ that gets sorted into this leaf.

    Consider $x\in \mathbb{R}^d$ such that $||x-r||_\infty \geq \varepsilon$ for all $r\in S_{audit}$. Say $\mathcal{T}(x)=1$, i.e., there exists a leaf such that $x$ belongs to this leaf and the leaf corresponds to class one. Consider $r \in S_{audit}$ that belongs to this leaf (by above, such $r$ exists). By the definition of the L-infinity norm there exists some dimension $i$ where $|r_i-x_i| > \varepsilon$. Suppose $r_i-x_i > \varepsilon$. Notice that there is a point $r - \varepsilon\vec{b}_i \in S_{train}'$ which satisfies that $(r-\varepsilon\vec{b}_i)_j=r_j$ for all $j\neq i$, and where $x_i < (r-\varepsilon \vec{b}_i)_i < r_i$. Then by above, $r-\varepsilon \vec{b}_i$ must be sorted into the same leaf as $r$ and $x$. But $r-\varepsilon \vec{b}_i$ has label $g(r)=0$, while for $x$ holds $\mathcal{T}(x)=1$. Thus, we found a contradition. The same argument holds if $x_i-r_i>\varepsilon$, but using the point $r+\varepsilon \vec{b}_i$ instead of $r-\varepsilon 
    \vec{b}_i$.
\end{proof}

\iffull
\section{Attack Detection}\label{sec:detection}

While proof of training alone cannot detect the attack above (as it relies on training the decision tree entirely honestly), nor can a black-box audit where the model owner knows the audit data before training time, we might still hope to detect when these attacks occur. For example, we might hope to conduct statistical tests on the training data to determine if it was honestly sampled from the underlying distribution or if it was adversarially constructed. In such a case, we cannot directly compare the training data to the true distribution of real data because the underlying distribution is not fully known to the auditor. Instead, we must compare the training data with a sample from that distribution. In the most simple case, this sample is the reference set $S_{audit}$.

We argue that under a certain family of functions, our constructed training set is indistinguishable from $S_{audit}$.


\begin{defn}
    Suppose $\vec{\alpha}$ is a set of bins over $d$ dimensions. Then $H_{\vec{\alpha}}:(\mathbb{R}^d\times \{ 0,1 \})^*\to \mathcal{H}$ is the function which takes databases over $d$ features and a binary classification to their normalized histogram with bins $\vec{\alpha}$.
\end{defn}

\begin{defn}
    A function $f:(\mathbb{R}^d\times \{0,1\})^*\to \mathbb{R}$ is called ($\gamma,c$)-magnitude insensitive if there exists a choice of bins $\vec{\alpha}$ and function $f':\mathcal{H}\to \mathbb{R}$ such that $|f(D)-f'(H_{\vec{\alpha}}(D))| < \gamma$ for all $D\in (\mathbb{R}^d\times \{0,1\})^*$ and $|f'(H_{\vec{\alpha}}(D))-f'(H_{\vec{\alpha}}(D||r))|\leq \frac{c}{|D|}$ for all $D\in(\mathbb{R}^d\times \{0,1\})^*$ and $r\in\mathbb{R}^d\times \{0,1\}$. 
\end{defn}

\begin{thm} \label{thm:metric-approx}
    If $f$ is ($\gamma,c$)-magnitude insensitive, then $|f(S_{audit})-f\left( S_{audit}^k||\delta \right)|\leq \varepsilon$ for any $\varepsilon>2\gamma$ and $k\geq \frac{2dc}{\varepsilon - 2\gamma}$, where $\delta$ is the additional training data created by Algorithm \ref{alg:attack} when run with input $S_{audit},d,\varepsilon,g$ for any $g$.
\end{thm}
\begin{proof}
    We will write $f'$ to be the $\gamma$-approximation of $f$ guaranteed to exist by the fact that $f$ is ($\gamma,c$)-magnitude insensitive. Observe that because $H_{\vec{\alpha}}$ takes databases to their normalized histograms, $H_{\vec{\alpha}}(S_{audit})=H_{\vec{\alpha}}\left( S_{audit}^k \right)$, because the non-normalized histograms of the two databases are simply scaled versions of one another.

Next, it will be helpful to show that for any two databases $D_{1},D_{2}\in (\mathbb{R}^d\times \{ 0,1 \})^*$, we have $|f'(H_{\vec{\alpha}}(D_{1}))-f'(H_{\vec{\alpha}}(D_{1}||D_{2}))| \leq c\frac{|D_{2}|}{|D_{1}|}$. Let us write $D_{2}=d_{1}||d_{2}||\ldots||d_{|D_{2}|}$. Then we get that 
\begin{equation*}
\begin{split}
&\mathrel{\phantom{=}}|f'(H_{\vec{\alpha}}(D_{1}))-f'(H_{\vec{\alpha}}(D_{1}||D_{2}))|\\
&= |f'(H_{\vec{\alpha}}(D_{1}))-f'(H_{\vec{\alpha}}(D_{1}||d_{1}))+f'(H_{\vec{\alpha}}(D_{1}||d_{1}))-\ldots\\&\qquad+f'(H_{\vec{\alpha}}(D_{1}||d_{1}||d_{2}||\ldots||d_{|D_{2}|-1}))-f'(H_{\vec{\alpha}}(D_{1}||D_{2}))|\\
&\leq |f'(H_{\vec{\alpha}}(D_{1}))-f'(H_{\vec{\alpha}}(D_{1}||d_{1}))|+\\&\qquad|f'(H_{\vec{\alpha}}(D_{1}||d_{1}))-f'(H_{\vec{\alpha}}(D_{1}||d_{1}||d_{2})|+\ldots\\&\qquad+|f'(H_{\vec{\alpha}}(D_{1}||d_{1}||d_{2}||\ldots||d_{|D_{2}|-1}))-f'(H_{\vec{\alpha}}(D_{1}||D_{2}))|\\
 &\leq \frac{c}{|D_{1}|}+\frac{c}{|D_{1}|+1} + \ldots + \frac{c}{|D_{1}|+|D_{2}|-1}\\
 &\leq c\frac{|D_{2}|}{|D_{1}|}
\end{split}
\end{equation*}
Then we can apply this to $S_{audit}^{k}$ and $S_{audit}^{k}||\delta$; recall that $|\delta| = 2d|S_{audit}|$. Then we see that \begin{align*}
    &\mathrel{\phantom{=}} \left|f'\left( H_{\vec{\alpha}}\left( S_{audit} \right) \right)-f'\left( H_{\vec{\alpha}}\left( S_{audit}^{k}||\delta \right) \right)\right|\\
    &=\left|f'\left( H_{\vec{\alpha}}\left( S_{audit}^{k} \right) \right)-f'\left( H_{\vec{\alpha}}\left( S_{audit}^{k}||\delta \right) \right)\right|\\&\leq c \frac{2d|S_{audit}|}{k|S_{audit}|}\\&\leq c \frac{2d}{\left( \frac{2dc}{\varepsilon-2\gamma} \right)}=\varepsilon-2\gamma
\end{align*} We have two cases now.

Case 1: $f'\left( H_{\vec{\alpha}}\left( S_{audit} \right) \right)\geq f'\left( H_{\vec{\alpha}}\left( S_{audit}^{k}||\delta \right) \right)$. Then we have \begin{equation*}\begin{split}
\varepsilon-2\gamma & \geq f'\left( H_{\vec{\alpha}}\left( S_{audit} \right) \right)-f'\left( H_{\vec{\alpha}}\left( S_{audit}^{k}||\delta \right) \right)\\
&= f(S_{audit})-f(S_{audit})+f'(H_{\vec{\alpha}}(S_{audit}))\\&\qquad-f(S_{audit}^{k}||\delta)+f(S_{audit}^{k}||\delta)-f'(H_{\vec{\alpha}}(S_{audit}^{k}||\delta))\\
& \geq f(S_{audit}) - |f(S_{audit})-f'(H_{\vec{\alpha}}(S_{audit}))|\\&\qquad - f(S_{audit}^{k}||\delta)-|f(S_{audit}^{k}||\delta)-f'(H_{\vec{\alpha}}(S_{audit}^{k}||\delta))|\\
& \geq f(S_{audit}) - \gamma - f(S_{audit}^{k}||\delta) - \gamma
\end{split}\end{equation*}
and so we see that $\varepsilon \geq f(S_{audit})-f(S_{audit}^{k}||\delta)$. We also have
\begin{equation*}\begin{split}
&\mathrel{\phantom{=}}f(S_{audit})-f(S_{audit}^{k}||\delta) \\
& = f'(H_{\vec{\alpha}}(S_{audit}))-f'(H_{\vec{\alpha}}(S_{audit}))+f(S_{audit})\\&\qquad-f'(H_{\vec{\alpha}}(S_{audit}^{k}||\delta))+f'(H_{\vec{\alpha}}(S_{audit}^{k}||\delta))-f(S_{audit}^{k}||\delta)\\
& \geq f'(H_{\vec{\alpha}}(S_{audit}))-|f'(H_{\vec{\alpha}}(S_{audit}))-f(S_{audit})|\\&\qquad-f'(H_{\vec{\alpha}}(S_{audit}^{k}||\delta))-|f'(H_{\vec{\alpha}}(S_{audit}^{k}||\delta))-f(S_{audit}^{k}||\delta)|\\
& \geq f'(H_{\vec{\alpha}}(S_{audit}))-\gamma-f'(H_{\vec{\alpha}}(S_{audit}^{k}||\delta))-\gamma\\
& \geq -2\gamma\\
& > -\varepsilon
\end{split}\end{equation*}
Then $|f(S_{audit})-f(S_{audit}^{k}||\delta)| \leq \varepsilon$.

Case 2: $f'(H_{\vec{\alpha}}(S_{audit})) \leq f'(H_{\vec{\alpha}}(S_{audit}^{k}||\delta))$. Then we have\begin{equation*}\begin{split}
\varepsilon-2\gamma & \geq f'\left( H_{\vec{\alpha}}\left( S_{audit}^{k}||\delta \right) \right)-f'\left( H_{\vec{\alpha}}\left( S_{audit} \right) \right)\\
&= f(S_{audit}^{k}||\delta)-f(S_{audit}^{k}||\delta)+f'(H_{\vec{\alpha}}(S_{audit}^{k}||\delta))\\&\qquad-f(S_{audit})+f(S_{audit})-f'(H_{\vec{\alpha}}(S_{audit}))\\
& \geq f(S_{audit}^{k}||\delta)-|f(S_{audit}^{k}||\delta)-f'(H_{\vec{\alpha}}(S_{audit}^{k}||\delta))|\\& \qquad - f(S_{audit}) - |f(S_{audit})-f'(H_{\vec{\alpha}}(S_{audit}))|\\
& \geq f(S_{audit}^{k}||\delta) - \gamma-f(S_{audit}) - \gamma
\end{split}\end{equation*}
and so we see that $\varepsilon \geq f(S_{audit}^{k}||\delta)-f(S_{audit})$. We also have
\begin{equation*}\begin{split}
& \mathrel{\phantom{=}}f(S_{audit}^{k}||\delta)-f(S_{audit})\\
& = f'(H_{\vec{\alpha}}(S_{audit}^{k}||\delta))-f'(H_{\vec{\alpha}}(S_{audit}^{k}||\delta))+f(S_{audit}^{k}||\delta)\\&\qquad\qquad -f'(H_{\vec{\alpha}}(S_{audit}))+f'(H_{\vec{\alpha}}(S_{audit}))-f(S_{audit})\\
& \geq f'(H_{\vec{\alpha}}(S_{audit}^{k}||\delta))-|f'(H_{\vec{\alpha}}(S_{audit}^{k}||\delta))+f(S_{audit}^{k}||\delta)|\\&\qquad\qquad-f'(H_{\vec{\alpha}}(S_{audit}))-|f'(H_{\vec{\alpha}}(S_{audit}))-f(S_{audit})|\\
& \geq f'(H_{\vec{\alpha}}(S_{audit}^{k}||\delta))-\gamma-f'(H_{\vec{\alpha}}(S_{audit}))-\gamma\\
& \geq -2\gamma\\
& \geq -\varepsilon
\end{split}\end{equation*}

Then $|f(S_{audit})-f(S_{audit}^{k}||\delta)| \leq \varepsilon$.
\end{proof}

This theorem does not suggest that it is completely impossible to detect the attack given in Algorithm \ref{alg:attack}. Rather, it only precludes detection by a certain class of functions. However, we argue that this class is expansive and covers many intuitive approaches.

The sole requirement for the audit metric $f$ is that it must be approximable by $f'$ which satisfies three properties. Firstly, $f'$ operates over histograms for some choice of bins $\vec{\alpha}$. This is a necessary condition, as if $f$ were not approximable by a function over a binning of the training data, we could drastically change the audit outcome by simply adding a small amount of noise to the data. Next, $f'$ must be relatively insensitive to additional data. The intuition here is that no individual datapoint should dramatically change the outcome of the audit. Finally, $f'$ operates over normalized histograms. This property is necessary for the proof to go through, but is satisfied by many intuitive audit metrics. For example, the mean and standard deviation of a feature (even conditioned on any arbitrary set of features) are approximable from a normalized histogram. 

\begin{lemma}\label{lem:mean-est}
    Let $\mu_{j}(D)$ be the mean of (bounded) feature $j$ of a dataset $D$. Then for every $\gamma>0$, $\mu_j(D)$ is $(\gamma,M-m)$-magnitude insensitive, where $B$ is the set of bins in the histogram and $M, m$ are an upper and lower bound on possible $j$-values respectively.
\end{lemma}
\begin{proof}
    Notice that $\mu_{j}(D) \approx \sum_{i\in B} p_{i}x_{j,i}$ where $B$ is the set of bins in the histogram, $p_{i}$ is the height of bin $i$ in the normalized histogram of $D$, and $x_{j,i}$ is the $j$-value of bin $i$. Let us show that for any $\gamma>0$, there exists a binning of the data such that this is a $\gamma$-approximation of $\mu_{j}(D)$. Let the bins in feature $j$ have width $\gamma$. Then for each datapoint $d$ with $j$ value $j_{d}$, bin $i$, and binned $j$-value $x_{j,i}$, we have that $|x_{j,i}-j_{d}|\leq \gamma$. Then \begin{equation*}\begin{split}
\sum_{i\in B}p_{i}x_{j,i} & = \sum_{i\in B} \frac{c_{i}}{|D|}x_{j,i}\\
& = \sum_{d\in D} \frac{1}{|D|}x_{j,i}\\
\implies\left|\sum_{i\in B}p_{i}x_{j,i} - \sum_{d\in D} \frac{1}{|D|}j_{d}\right| & = \left|\sum_{d\in D} \frac{1}{|D|}x_{j,i} - \sum_{d\in D} \frac{1}{|D|}j_{d}\right|\\
& = \left|\frac{1}{|D|}\sum_{d\in D} (x_{j,i} - j_{d})\right|\\
& \leq \frac{1}{|D|}\sum_{d\in D} \left|x_{j,i} - j_{d}\right|\\
& \leq \frac{1}{|D|}\sum_{d\in D} \gamma\\
& = \gamma\\
\end{split}
\end{equation*}
Next, let us show that the sensitivity of our approximation of $\mu_{j}$ is upper bounded by $\frac{M-m}{|D|}$. Notice that by adding a single point, one histogram bin will increase by 1 and the rest will be unchanged. Then for every bin $k$, \begin{equation*}\begin{split}
& \sum_{i\in B} \frac{c_{i}}{|D|+1}x_{j,i} + \frac{1}{|D|+1}x_{j,k}- \sum_{i\in B} \frac{c_{i}}{|D|}x_{j,i}\\
& \quad= \sum_{i\in B} c_i x_{j,i} \left(\frac{1}{|D|+1}-\frac{1}{|D|}\right) + \frac{x_{j,k}}{|D|+1}\\
& \quad= -\left(\sum_{i\in B} \frac{c_i x_{j,i}}{|D|^2+|D|}\right) + \frac{x_{j,k}}{|D|+1}\\
& \quad\leq - \left(\frac{m}{|D|+1}\right)+\frac{M}{|D|+1}\\
& \quad\leq \frac{M-m}{|D|}\\
& \sum_{i\in B} \frac{c_{j}}{|D|+1}x_{j,i} + \frac{1}{|D|+1}x_{j,k}- \sum_{i\in B} \frac{c_{j}}{|D|}x_{j,i}\\
& \quad= -\left(\sum_{i\in B} \frac{c_i x_{j,i}}{|D|^2+|D|}\right) + \frac{x_{j,k}}{|D|+1}\\
& \quad\geq -\left(\frac{M}{|D|+1}\right)+\frac{m}{|D|+1}\\
& \quad\geq \frac{m-M}{|D|}\\
\end{split}\end{equation*}
So we have that the sensitivity is no greater than $\frac{M-m}{|D|}$.
\end{proof}

We will proceed to use this fact to show that Welch's $t$-test will fail to detect this attack.

\begin{cor} \label{cor:welch-t-test}
    (Restatement of Corollary~\ref{cor:welch-main}) Given an audit dataset $S_{audit}$ and significance level $\alpha$, we can use Algorithm \ref{alg:attack} to construct a training dataset $S_{train}'$ such that for any feature $j$, $S_{train}'$ passes Welch's t-test when its values in feature $j$ are compared to those of $S_{audit}$ with significance level $\alpha$. 
\end{cor}
Before we can prove this corollary, we will need a lemma which bounds the concentration of the Student's $t$-distribution.

\begin{lemma}\label{lem:t-dist-bound}
    If $X$ and $Z$ are random variables drawn independently from the Student's t-distribution with $\nu$ degrees of freedom and the standard normal distribution respectively, then for every $t>0$, we have
\begin{equation*}
    \Pr[|X|< t] \leq \Pr[|Z|<t]
\end{equation*}
\end{lemma}

\begin{proof}
    We will write $F_{X}(t)$ to denote the CDF of random variable $X$ evaluated at $t$, and $f_{X}(t)$ the PDF. We will also write $\mathbb{E}_{X}(g(X))$ to be the expected value of $g(X)$ with randomness over $X$. Let us begin by demonstrating that for all $t<0$, we have $F_{X}(t) > F_{Z}(t)$. First, recall that if $W$ and $Y$ are drawn from the $\chi^2$ distribution with $\nu$ degrees of freedom and the standard normal distribution respectively, then $Y\sqrt{ \frac{\nu}{W} }$ is distributed according to the Student's $t$-distribution with $\nu$ degrees of freedom, so let us write $X=Y\sqrt{  \frac{\nu}{W} }$. Then according to the law of total probability, we have
\begin{equation*}\begin{split}
F_{X}(t) & = \int_{0}^\infty F_{Y}\left( t\sqrt{ \frac{w}{\nu} } \right) f(w) dw\\
& = \mathbb{E}_{W}\left( F_{Y}\left( t\sqrt{ \frac{W}{\nu} } \right) \right)
\end{split}\end{equation*}
Notice that $\frac{d^2}{dt^2}F_{Y}(t) = \frac{d}{dt}f_{Y}(t)=\frac{d}{dt} \frac{1}{\sqrt{ 2\pi }}e^{-\frac{t^2}{2}}=-\frac{t}{\sqrt{ 2\pi }}e^{-\frac{t^2}{2}}>0$ when $t<0$. Then since $t\sqrt{ \frac{W}{\nu} }$ must be less than $0$, we can apply Jensen's inequality to get
\begin{equation*}\begin{split}
F_{X}(t) & = \mathbb{E}_{W}\left( F_{Y}\left( t\sqrt{ \frac{W}{\nu} } \right) \right)\\
& \geq F_{Y}\left( \mathbb{E}_{W}\left( t\sqrt{ \frac{W}{\nu} } \right) \right)\\
& = F_{Y}\left( t\mathbb{E}_{W}\left(\sqrt{ \frac{W}{\nu} } \right) \right)
\end{split}\end{equation*}
Then since $\frac{d^2}{du^2}\sqrt{ u }=-\frac{1}{4\sqrt{ u^3 }}\leq 0$, we get that $\mathbb{E}_{W}\left( \sqrt{ \frac{W}{\nu} } \right) \leq \sqrt{ \frac{\mathbb{E}_{W}(W)}{\nu} }=\sqrt{ \frac{\nu}{\nu} }=1$. So because $t<0$, we can see that $t\mathbb{E}_{W}\left( \sqrt{ \frac{W}{\nu} } \right) \geq t$, and since $F_{Y}(u)$ is increasing, we get
\begin{equation*}\begin{split}
F_{X}(t) &\geq F_{Y}\left( t\mathbb{E}_{W}\left( \sqrt{ \frac{W}{\nu} } \right) \right)\\
& \geq F_{Y}\left( t \right)\\
\end{split}\end{equation*}
Since $f_{X}$ and $f_{Y}$ are both symmetric about $t=0$, it then follows by a symmetric argument that for all $t>0$, $F_{X}(t) \leq F_{Y}(t)$. Then we see that for any $t>0$, \begin{equation*}\begin{split}
\Pr[|X|<t] & = F_{X}(t)-F_{X}(-t)\\
& \leq F_{Y}(t) - F_{Y}(-t)\\
& = \Pr[|Y|<t]\\
& = \Pr[|Z|<t]\\
\end{split} 
\end{equation*}
Because $Y$ and $Z$ are independently and identically distributed.
\end{proof}

We are now ready to prove Corollary \ref{cor:welch-t-test}.
\begin{proof}[Proof of Corollary \ref{cor:welch-t-test}]
    A pair of datasets $D_1,D_2$ pass Welch's $t$-test on feature $j$ if \begin{equation*}
        \frac{|\mu_j(D_1)-\mu_j(D_2)|}{\sqrt{\frac{\sigma_1^2}{|D_1|}+\frac{\sigma_2^2}{|D_2|}}} \leq T_{\alpha,\nu}
    \end{equation*}
    where $\alpha$ is the desired significance level, $\nu$ is the degrees of freedom in the datasets, and $T_{\alpha,\nu}$ is the unique value such that \begin{equation*}
        \Pr_{x\sim t(\nu)}[|x|\geq T_{\alpha,\nu}]=\alpha
    \end{equation*}
    where $t(\nu)$ is the Student's $t$-distribution with $\nu$ degrees of freedom. In our case, the $t$-test compares the reference dataset $S_{audit}$ with the training dataset $S_{train}'$.

    The value of $\nu$, and thus the value of $T_{\alpha,\nu}$, depends on the size of the datasets, with the threshold $T_{\alpha,\nu}$ decreasing as the datasets grow large. However, we will use Lemma \ref{lem:t-dist-bound} to give a lower bound for $T_{\alpha,\nu}$ which is constant with respect to $|S_{train}'|$. Then, we will show that by Lemma \ref{lem:mean-est} and Theorem \ref{thm:metric-approx} we can use Algorithm \ref{alg:attack} to construct a malicious training dataset $S_{train}'$ which maintains an arbitrarily small test statistic, and in particular, a dataset such that the test statistic is below the lower bound on the threshold.

    First, let us establish a lower bound on $T_{\alpha,\nu}$. Let us define $T'_{\alpha}$ to be the unique positive value such that \begin{equation*}
        \Pr_{Z\sim \mathcal{N}(0,1)}[|Z|\geq T'_\alpha] = \alpha
    \end{equation*}
    Then recall that Lemma \ref{lem:t-dist-bound} gives us that \begin{equation*}
        \Pr_{X\sim t(\nu)}[|X|<T'_\alpha] \leq \Pr_{Z\sim \mathcal{N}(0,1)}[|Z| < T'_\alpha]
    \end{equation*}
    If we write $f_X$ and $f_Z$ to represent the probability density functions (PDFs) of $X$ and $Z$ respectively, then we get equivalently that \begin{equation*}
        \int_{-T'_{\alpha}}^{T'_{\alpha}} f_X(u)du \leq \int_{-T'_{\alpha}}^{T'_{\alpha}} f_Z(u)du
    \end{equation*}
    Then we see that \begin{equation*}
        \begin{split}
            \Pr_{Z\sim \mathcal{N}(0,1)}[|Z|\geq T'_\alpha] & = \Pr_{X\sim t(\nu)}[|X|\geq T_{\alpha,\nu}]\\
            \implies \int_{-T'_\alpha}^{T'_\alpha} f_Z(u)du & = \int_{-T_{\alpha,\nu}}^{T_{\alpha,\nu}}f_X(u)du\\
            & = \int_{-T_{\alpha,\nu}}^{-T'_{\alpha}}f_X(u)du+\int_{-T'_{\alpha}}^{T'_{\alpha}}f_X(u)du\\&\qquad+\int_{T'_{\alpha}}^{T_{\alpha,\nu}}f_X(u)du\\
            & \leq \int_{-T_{\alpha,\nu}}^{-T'_{\alpha}}f_X(u)du+\int_{-T'_{\alpha}}^{T'_{\alpha}}f_Z(u)du\\&\qquad+\int_{T'_{\alpha}}^{T_{\alpha,\nu}}f_X(u)du\\
            \implies 0 & \leq \int_{-T_{\alpha,\nu}}^{-T'_{\alpha}}f_X(u)du+\int_{T'_{\alpha}}^{T_{\alpha,\nu}}f_X(u)du\\
        \end{split}
    \end{equation*}
    Then because $f_X(x)$ is symmetric about $x=0$, this yields \begin{equation*}
        2\int_{T'_{\alpha}}^{T_{\alpha,\nu}}f_X(u)du \geq 0
    \end{equation*}and thus\begin{equation*}
        \int_{T'_{\alpha}}^{T_{\alpha,\nu}}f_X(u)du \geq 0
    \end{equation*}
    Now recall the simple result from calculus that states that if $g$ is positive valued, then \begin{equation*}
        \int_a^b g(x)dx \geq 0 \iff a\leq b
    \end{equation*}
    Then because $f_X$ is positive-valued, our prior result entails that $T_{\alpha,\nu} \geq T'_\alpha$, so $T_\alpha'$ is a lower bound on $T_{\alpha,\nu}$ that does not depend on $|S_{train}'|$.

    Next, observe that the test statistic for Welch's $t$-test has the following upper bound: \begin{equation*}
        \frac{|\mu_j(S_{train}')-\mu_j(S_{audit})|}{\sqrt{\frac{\sigma_{train}^2}{|S_{train}'|}+\frac{\sigma_{audit}^2}{|S_{audit}|}}} \leq \frac{|\mu_j(S_{train}')-\mu_j(S_{audit})|}{\sqrt{\frac{\sigma_{audit}^2}{|S_{audit}|}}}
    \end{equation*}
    Furthermore, Lemma \ref{lem:mean-est} implies that for any $\varepsilon>0$, we can choose $\gamma < \frac{\varepsilon}{2}$ such that $\mu_j$ is $(\gamma,c)$-magnitude insensitive, and so by Theorem \ref{thm:metric-approx}, Algorithm \ref{alg:attack} yields a dataset $S_{train}'$ such that $|\mu_j(S_{train}')-\mu_j(S_{audit})|\leq \varepsilon$ when appropriately parameterized. 
    Then let $\varepsilon = T'_\alpha\frac{\sigma_{audit}}{2\sqrt{|S_{audit}|}}$. This produces the result that \begin{equation*}
        \begin{split}
            \frac{|\mu_j(S_{train}')-\mu_j(S_{audit})|}{\sqrt{\frac{\sigma_{train}^2}{|S_{train}'|}+\frac{\sigma_{audit}^2}{|S_{audit}|}}} &\leq \frac{2\varepsilon}{\sqrt{\frac{\sigma_{audit}^2}{|S_{audit}|}}}\\
            & = \frac{2}{\sqrt{\frac{\sigma_{audit}^2}{|S_{audit}|}}}T'_\alpha\frac{\sigma_{audit}}{2\sqrt{|S_{audit}|}}\\
            & = T'_\alpha\\
            & \leq T_{\alpha,\nu}
        \end{split}
    \end{equation*}
    which passes the $t$-test for feature $j$. Finally, by choosing $k=\max_j \frac{4d(M_j-m_j)\sqrt{|S_{audit}|}}{T'_\alpha \sigma_{audit,j}}$ we get for every feature $i$ that $|\mu_i(S_{train}')-\mu_i(S_{audit})|\leq 2\min\limits_{j} T_\alpha' \frac{\sigma_{audit,j}}{2\sqrt{|S_{audit}|}}\leq 2 T'_\alpha \frac{\sigma_{audit,i}}{2\sqrt{|S_{audit}|}}$, so $S_{train}'$ passes the $t$-test for feature $i$.
\end{proof}

\fi
\section{Case Study}\label{sec:case_study}
We now discuss a number of state of the art works that consider the problem of privacy-preserving auditing.
These works are focused on different auditing functions (accuracy, fairness, etc), different types of machine learning models, and their security models they use are not necessarily aligned. We now briefly outline the techniques and security guarantees that are claimed in each of the works. Our goal is not to provide an exhaustive survey, but rather to illustrate the landscape through recent works that are broadly representative of the field—even though they span different years, venues, and communities (ranging from machine learning to security).

\subsection{zkCNN: Zero Knowledge Proofs for Convolutional Neural Network Predictions and Accuracy}
\noindent\textbf{Goal and Solution Details.}
\cite{LiuXZ21} propose zkCNN, a zero-knowledge proof protocol for inference and accuracy of convolutional neural networks (CNNs). The core contribution is a novel sumcheck protocol (which is the key ingredient in many zero-knowledge system) that is tailored to two-dimensional convolutions.

\noindent\textbf{Security Model.} zkCNN considers the standard setting with a prover and a verifier.
Either party can be malicious. \cite{LiuXZ21}'s security definition for inference is a zero-knowledge-style definition, and the scheme is required to satisfy correctness, soundness, and zero-knowledge. Similar to \cite{ZhangFZS20}, \cite{LiuXZ21}'s soundness intuitively states that a prover should not be able to output a commitment to a model and provide a proof $\pi$, prediction $y$ and datapoint $X$ such that the verifier accepts the proof, and at the same time, the committed model's prediction for $X$ is not equal to $y$. If instantiated with a specific commitment scheme, \cite{LiuXZ21}'s scheme further satisfies knowledge soundness, the stronger version of soundness where there exists an extractor to extract the CNN parameters from a valid proof and prediction with overwhelming probability. \cite{LiuXZ21} do not provide a security definition for their proof of accuracy.

\noindent\textbf{Discussion.}
As \cite{LiuXZ21} do not give a security definition for their proof of accuracy, the formal security guarantee they provide is not fully clear. However, the authors indicate that their scheme can be used to prove the accuracy on a public dataset. This scenario falls within our framework of definition~\ref{def:forge-audit-new}, and is vulnerable to the same style of attack as outlined in \S\ref{sec:attack-methods}.

\noindent\textbf{Expanding/Restoring Security.} Similar to Zhang et al.~\cite{ZhangFZS20}, \cite{LiuXZ21}'s construction can be adapted using our techniques from \S~\ref{sec:auditcsp-accuracy-overview} to satisfy the guarantee with respect to new datasets drawn from the same distribution. In particular, Liu et al.'s scheme can be plugged in as the ZKP of our protocol $\auditcsp$ in \S\ref{sec:construction-theorem-overview}.

\subsection{Confidential-PROFITT: Confidential PROof of FaIr Training of Trees}
\noindent\textbf{Goal and Solution Details.} Shamsabadi et al.~\cite{shamsabadi2022profitt} propose Confidential-PROFITT, a framework for certifying fairness of decision trees while preserving confidentiality of both the model and the training data. Confidential-PROFITT consists of a zero-knowledge-friendly decision tree learning algorithm that, when executed honestly, enforces fairness by design—up to a tunable degree controlled by a parameter. On top of this, Confidential-PROFITT designs a zero-knowledge proof system to verify fairness of a decision tree. The proof requires the model provider to commit to both the model and its training data, then prove in zero-knowledge that the paths taken by the committed training points through the (committed) decision tree satisfy specified fairness bounds. In terms of fairness metrics, Confidential-PROFITT supports demographic parity and equalized odds.

\noindent\textbf{Security Model.} Confidential-PROFITT considers a malicious model provider (that, however, is assumed to commit to the training data honestly) and
a malicious auditor, and obtains standard zero-knowledge proof properties (correctness, soundness, zero-knowledge) with respect to a statement that can be summarized roughly as follows ``With respect to a private dataset \emph{chosen by the model provider}, the committed model satisfies certain fairness guarantees''. 

\noindent\textbf{Discussion.} 
Confidential-PROFITT assumes that the model provider honestly commits to the training data. Under this assumption, the corresponding zero-knowledge proof certifies that the resulting model inherits the fairness guarantees of the fair learning algorithm introduced in Confidential-PROFITT (which the authors show indeed improves fairness). However, if the provider is not restricted to committing to the true training data, Confidential-PROFITT is vulnerable to data-forging attacks, as the provider can choose the audit dataset before committing to the model. 

\noindent\textbf{Expanding/Restoring Security.} We believe Confidential-PROFITT's zero-knowledge proof of fairness could potentially be combined with our techniques from \S~\ref{sec:auditcsp-accuracy-overview} if one would swap the training data for the data sampled by the auditor. However, the adaptation is not as straight-forward as that for Zhang et al.~\cite{ZhangFZS20} and Liu et al.~\cite{LiuXZ21}. This is because Confidential-PROFITT verifies not only the actual fairness guarantee, but also additional fairness checks for intermediate nodes in the decision tree. A careful analysis is required to ensure that the generalization of  fairness properties \iffull(similar to what we show for demographic parity in \S~\ref{sec:auditcsp-demographic-parity})\fi also extends to intermediate nodes. We leave this as an interesting
direction for future work.

\subsection{P2NIA: Privacy-Preserving Non-Iterative Auditing}
\noindent\textbf{Goal and Solution Details.} Bourrée et al.~\cite{p2nia} propose a novel auditing scheme that enables one-shot verification of a model’s group fairness while preserving privacy for both parties: the model provider is not required to open-source the model, and the auditor need not disclose any private information to support the audit. The main contribution of \cite{p2nia} is a mechanism that enables auditing without requiring the auditor to supply the audit dataset. Specifically, the model provider supplies a dataset together with the corresponding predictions (both in the clear), which the auditor then uses to verify the fairness condition. To construct this dataset, model provider draws on a portion of its internal training data. To preserve confidentiality of this data, it is not shared directly. Instead, model provider feeds it into a synthetic data generation algorithm, and the resulting synthetic dataset is what is sent to the auditor. 

\noindent\textbf{Security Model.} The work does not provide a formal security model. It is set up in the black-box setting and assumes that the auditor does not know the distribution of the model owner's training data.

\noindent\textbf{Discussion.} As \cite{p2nia} do not utilize cryptographic techniques to prove that the outputs actually correspond to the given inputs, the prover can easily cheat by simply adjusting the labels it supplies for the constructed dataset. However, even if one were to strengthen the scheme by adding a secure proof of training (e.g., \cite{pappas2024sparrow}) together with inference proofs (as in \cite{ZhangFZS20}), the fact that the model owner knows the dataset that is being used for the audit means that the solution falls within our framework of Definition~\ref{def:forge-audit-new}, and is thus vulnerable to data-forging attacks. An interesting open question would be to see if, since in this scenario the model owner not only knows, but directly influences the audit dataset, there can be an even simpler attack.

\subsection{OATH: Efficient and Flexible Zero-Knowledge Proofs of End-to-End ML Fairness}
\noindent\textbf{Goal and Solution Details.} \cite{franzese2024oath} present OATH, a model-agnostic fairness auditing framework. The core idea in OATH is to leverage clients (who query the model during deployment) to participate in the auditing process. OATH operates in two phases: (i) a certification protocol between the model provider and the auditor, and (ii) a query authentication protocol involving model provider, inference clients, and auditor (dubbed verifier in OATH). The first phase follows the standard certification flow we describe in \S\ref{sec:mlback}. In the second phase, the auditor receives commitments to client queries and the corresponding model predictions. These commitments can later be verified in zero knowledge for fairness, correctness, and consistency with the certified model.

\noindent\textbf{Security Model.} OATH considers three fully malicious entities: a model provider, inference clients, and an auditor. These parties are assumed not to collude with each other. The auditor assesses model fairness both with respect to the calibration dataset and the clients queries. The system provides standard correctness, soundness, and zero-knowledge with respect to these two datasets. 

\noindent\textbf{Discussion.} The calibration dataset which is used in the certification protocol between the model provider and the auditor might be supplied by either party. If the calibration dataset is chosen by the prover, same as P2NIA and Confidential-PROFITT, the corresponding fairness check is vulnerable to data forging. However, in contrast to prior works, OATH can fall back on guarantees based on client’s queries. 

\subsection{FairProof: Confidential and Certifiable Fairness for Neural Networks}
\noindent\textbf{Goal and Solution Details.} \cite{yadav2024fairproof} propose FairProof, a fairness certification approach that maintains confidentiality of the model. In contrast to Confidential-PROFITT and OATH, which focus on group fairness metrics, FairProof considers local individual fairness. This allows \cite{yadav2024fairproof} to issue a personalized certificate to every client. 

\noindent\textbf{Security Model.}
FairProof system involves a malicious model provider and malicious clients (who wish to learn model details/training data), and considers standard correctness, soundness, and zero-knowledge properties. The corresponding statement is roughly as follows: ``Given a datapoint $x$, the model's output is $y$ and a lower bound on an individual fairness parameter for $x$ is $\epsilon_x$''.

\noindent\textbf{Discussion.}
The usage of a specific fairness metric (local individual fairness) allows FairProof to provide per-client certificates of fairness, and escape the problems that arise from the usage of reference datasets (including vulnerability to data-forging attacks). On the flip side, FairProof requires to generate fairness certificates during deployment and does not provide any fairness guarantees prior to deployment.

\subsection{Scaling up Trustless DNN Inference with Zero-Knowledge Proofs}
\noindent\textbf{Goal and Solution Details.} Kang et al.~\cite{kangdnn} propose a zero-knowledge-based framework for verifying DNN inference and accuracy. Their key contribution is a careful translation of DNN specifications into arithmetic circuits suitable for zero-knowledge proofs. The system also introduces economic incentives to support ML-as-a-service. Concretely, when verifying accuracy, the model provider first commits to the model, and the client commits to the test set. Both parties then deposit monetary collateral into an escrow. The client reveals the test set, and the provider must produce a zero-knowledge proof that the committed model meets the claimed accuracy. If the provider fails or refuses to prove the required accuracy, it forfeits its collateral; otherwise, the client pays for the service.

\noindent\textbf{Security Model.}
Kang et al. study the standard two-party setting with a \emph{prover} (model provider) and a \emph{verifier} (client), either of whom may be malicious. Cryptographically, they aim for the standard zero-knowledge proof properties: \emph{completeness}, \emph{knowledge soundness}, and \emph{zero knowledge}. They further consider incentives, showing that—under certain assumptions—honest model providers and clients are motivated to participate in the accuracy verification protocol, while malicious parties are discouraged.

\noindent\textbf{Discussion.} In terms of cryptographic guarantees, Kang et al. gets the core design right: their protocol for proofs of accuracy closely follows the framework outlined in \S\ref{sec:sec-definition} and is not vulnerable to our data-forging attacks. However, \cite{kangdnn} provide no formal guarantees about accuracy on data outside the audited set.\iffull It would be interesting to perform an analysis similar to that in \S\ref{sec:auditcsp-accuracy} given their constraints.\fi

\subsection{ezDPS: An Efficient and Zero-Knowledge Machine Learning Inference Pipeline}
\noindent\textbf{Goal and Solution Details.} Wang and Hoang~\cite{WangH23} introduce ezDPS, a pipeline for zero-knowledge proofs of inference correctness and accuracy above a specified threshold. They construct arithmetic circuit gadgets for key ML operations, including exponentiation, absolute value, and array max/min, and further devise optimized methods for proving Discrete Wavelet Transform, Principal Component Analysis, and multi-class Support Vector Machines with various kernel functions using an efficient set of arithmetic constraints.

\noindent\textbf{Security Model.} \cite{WangH23} consider two mutually distrusting parties -- a malicious server and a semi-honest client, who follows the protocol but aims to learn information about the model's parameters. For their inference pipeline, they consider standard definitions of correctness, soundness, and zero-knowledge (similar to those by \cite{ZhangFZS20} and \cite{kangdnn}). \cite{WangH23} do not provide a security definition for their proof of accuracy.

\noindent\textbf{Discussion.}
Similar to Liu et al.~\cite{LiuXZ21}, as \cite{WangH23} do not provide a security definition for their proof of accuracy, the precise security guarantee they achieve is somewhat unclear. However, \cite{WangH23} indicate that their scheme can be used to prove the accuracy on a public dataset, which falls within our framework of definition~\ref{def:forge-audit-new}. This instantiation of their method is vulnerable to the same style of attack as outlined in \S\ref{sec:attack-methods}.

\noindent\textbf{Expanding/Restoring Security.} Similar to Zhang et al.~\cite{ZhangFZS20} and Liu et al.~\cite{LiuXZ21}, the scheme of Wang and Hoang can be adapted using our techniques from \S~\ref{sec:auditcsp-accuracy-overview} to satisfy the guarantee with respect to freshly sampled datasets. In particular, their protocol  can be plugged in as the ZKP of our protocol $\auditcsp$ in \S\ref{sec:construction-theorem-overview}.

\iffull
\section{Details of Secure CMC Framework}\label{sec:construction}

\subsection{Auditing Protocol and Security Definitions}\label{sec:securitydefs}
In this section, we formally define the security properties of an auditing protocol.

\noindent\textbf{Completeness}
An auditing protocol $\Pi$ is complete with error $p$ if for any model $h$ such that $F(h,\settrain)=1$, the following holds:	
\begin{align*}
	\probrv{b = 1}{(\com, b)\gets\interact{\prove(h,\settrain)}{\audit}} \geq 1-p
\end{align*}

\noindent\noindent\textbf{Binding}
An auditing protocol $\Pi$ is \emph{computationally binding} if for any PPT adversary $\cA$, the following is negligible in $\secp$:
\begin{align*}
    \probrv{\commit(h||\settrain;\rho) \\ = \commit(h'||\settrain';\rho') \\ \land\; (h \neq h' \lor \; \settrain\neq\settrain')}{(h,h',\settrain, \settrain', \rho,\rho')\gets\cA(1^\secp)}
\end{align*}
Note that we require the binding property to hold both for the model and the training dataset.
This can be easily achieved by separately committing to the model and the training dataset with a standard binding commitment scheme (\S~\ref{sec:comsecurity}), and then outputting the concatenation of the two commitments.

\medskip
\noindent\textbf{$\tilde{F}$-Relaxed Knowledge Soundness} 
An auditing protocol $\Pi$ is $\tilde{F}$-relaxed knowledge sound with knowledge error $\kappa$ if for any PPT adversary $\prove^*$, there exists an expected polynomial time extractor $\ext_{\prove^*}$ such that the following holds:
\begin{align*}
	\pext \geq \pacc - \kappa
\end{align*}
where 
\begin{align*}
	\pacc = \probrv{b=1}{(\com,b)\gets\interact{\prove^*}{\audit}}
\end{align*}
and $\pext$ is defined as
\begin{align*}
\Pr\left[\begin{aligned}
(\com = \commit(h||\settrain;\rho) &\land\ \tilde{F}(h,\settrain)= 1 ) \\
	                             &\lor \\
		\big(\com = \commit(h||\settrain;\rho) &= \commit(h'||\settrain',\rho') \\
		   & \land (h \neq h' \lor \settrain\neq\settrain') \big)
\end{aligned}\right]
\end{align*}
where $(\com,b)\gets\interact{\prove^*}{\audit}$ and $(h,\settrain,\rho,h',\settrain',\rho')\gets\ext_{\prove^*}(\com)$,

Intuitively, this notion guarantees that if a cheating prover $\prove^*$ convinces the auditor to accept with non-negligible probability, then it must either know a model $h$ and a training dataset $\settrain$ satisfying a predicate $\tilde{F}$, \textbf{or} find two distinct openings to the same commitment $\com$.
As the latter event happens with negligible probability if $\commit$ is computationally binding, this implies that $\prove^*$ must know a valid model $h$ and a training dataset $\settrain$ satisfying $\tilde{F}$.
We call this property ``$\tilde{F}$-relaxed'' knowledge soundness because the predicate $\tilde{F}$ is a relaxation of the original predicate $F$.
This is necessary because the auditor only checks the property $f$ on a \emph{finite sample}, which may not perfectly reflect the property $F$ on the \emph{underlying distribution}.
In our concrete instantiations, we will quantify the gap between $F$ and $\tilde{F}$ (see \S~\ref{sec:auditcsp-accuracy} and \S~\ref{sec:auditcsp-demographic-parity}).

\medskip
\noindent\textbf{Zero Knowledge} 
Let $\view_{\audit}^{\prove(h,\settrain)}$ be a string consisting of all the incoming messages that $\audit$ receives from $\prove$ during the interaction $\interact{\prove(h,\settrain)}{\audit}$, and $\audit$'s random coins.
$\Pi$ is \emph{zero-knowledge} against semi-honest auditor if there exists a PPT simulator $\simul$ such that for any PPT adversary $\cA$, and any $h$ such that $F(h,\settrain)=1$, the following is negligible in $\secp$.

\begin{align*}
	\left|
	\begin{aligned}
		&\probrv{b=1}{
			& b\gets\cA(\view_{\audit}^{\prove(h,\settrain)})}\\
		- &\probrv{b=1}{
			& \view' \gets\simul(1^\secp); 
			& b\gets\cA(\view')}
	\end{aligned}\right|
\end{align*}

\subsection{Construction of Auditing Protocol}\label{sec:constructiondetails}
We construct a commit-sample-prove auditing scheme $\auditcsp$. 
While we focus on a protocol checking $F$ on a hypothesis $h$ only, the construction below can be naturally extended to a more complex $F$ that additionally takes a training dataset as input. 
Let $\commit$ be a binding commitment scheme (\ref{sec:comsecurity}) and $\ZKP= (\cP,\cV)$ be a ZK proof system for the following relation $\rel$: for a pair of public statement $\stm = (\com,S_\text{audit})$ and private witness $\wit = (h,\rho)$, we have $(\stm,\wit)\in\rel \iff f(h,S_\text{audit})=1\land \com = \commit(h;\rho)$. 
We define a commit-sample-prove auditing protocol $\auditcsp=(\commit, \prove, \audit)$ using an emprical predicate $f$ and a distribution $\mathcal{D}$ over a query space $\querysp = \{(x_i,y_i)\}_{i=1}^m$ as follows:



\noindent\underline{$\interact{\prove(h)}{\audit}$}
\begin{enumerate}
    \item $\prove$ computes $\com=\commit(h;\rho)$ using a uniformly random string $\rho\in\{0,1\}^{l_{\commit}}$ and sends $\com$ to $\audit$.
    \item $\audit$ samples $S_\text{audit} \sim \mathcal{D}^n$ and sends it to $\prove$.
    \item $\prove$ and $\audit$ execute $b\gets\interact{\cP(\wit)}{\cV}(\stm)$, where $\stm = (\com,S_\text{audit})$ and $\wit = (h,\rho)$. Here, $\prove$ plays $\cP$ and $\audit$ plays $\cV$. \label{step:zkp-interaction}
    \item $\prove$ outputs $\com$, while $\audit$ outputs $b$.
\end{enumerate}

We now state our main theorem regarding the security of $\auditcsp$.
The result is stated for a general auditing task defined by a predicate $F$ and an empirical predicate $f$.

\begin{thm}\label{thm:auditcsp-general}
Suppose the empirical predicate $f$, the distributional predicate $F$, and the relaxed distributional predicate $\tilde{F}$ satisfy the following false negative and false positive rate bounds for every model $h$:
\begin{align*}
	& \Pr_{\setaudit \sim \mathcal{D}^n}[f(h,\setaudit)\neq 1 \mid F(h)=1] \leq \pfnr \\
	& \Pr_{\setaudit \sim \mathcal{D}^n}[f(h,\setaudit)=1 \mid \tilde{F}(h) \neq 1] \leq \pfpr
\end{align*}
Then $\auditcsp$ is a secure auditing protocol for $F$ satisfying the following properties: 
\begin{itemize}
	\item If $\ZKP$ is perfectly complete, then $\auditcsp$ is complete with error $\pfnr$ for any model $h$ such that $F(h)=1$.
	\item If the underlying commitment scheme $\commit$ is computationally binding, then $\auditcsp$ is computationally binding.
	\item If $\ZKP$ is knowledge sound with knowledge error $\kappa$, then $\auditcsp$ is $\tilde{F}$-relaxed knowledge sound with knowledge error $\kappa  + \pfpr$.
	\item If $\pfnr$ is negligible in $\secp$\footnote{We require this to realize negligible distinguishing advantage as in the standard zero-knowledge definition. If $\pfnr$ is small but not negligible, one could relax the zero-knowledge property for CMC to allow for non-negligible distinguishing advantage. This relaxation is analogous to knowledge soundness with potentially non-negligible knowledge error including $\pfpr$.}, $\commit$ is hiding, and $\ZKP$ is honest verifier zero-knowledge, then $\auditcsp$ is zero-knowledge against semi-honest auditor.
\end{itemize}
\end{thm}

\begin{proof} 
Binding trivially follows from the computational binding property of the underlying commitment scheme.

\noindent\underline{Completeness:}
By the assumption on $f$ and $F$, an honest auditing prover $\prove$ fails to convince the auditor playing verifier $\cV$ after receiving fresh $\setaudit$ with probability at most $\pfnr$.
Since $\commit$ and $\ZKP$ are perfectly complete, an honest $\ZKP$ prover $\cP$ given any valid witness always convinces the verifier. 
Thus, the overall completeness error is at most $\pfnr$.

\medskip
\noindent\underline{Knowledge Soundness:}
The protocol can be viewed as a commit-and-prove zero-knowledge proof with interleaved probabilistic tests on the statement. This approach is common in lattice-based zero-knowledge proofs (Cf. Theorem 5.1.6 of \cite{Nguyen}).

Let $\prove^*$ be any PPT adversary. 
We denote by $\cP_1$ the interactive ZK prover algorithm that $\prove^*$ invokes in Step \ref{step:zkp-interaction} to prove the statement $\stm$ fixed by the previous steps.  
Moreover, denote by $R_\cE$ (resp. $R_{\cP}$) the randomness space of the extractor $\cE$ (resp. prover $\cP_1$) for $\ZKP$.
We first construct the following extractor: 

\noindent\underline{$\ext_{\prove^*}$}
\begin{enumerate}
	\item Run $\prove^*$ to get $\com$.
	\item Sample $S \sim \mathcal{D}^n$, $r_\cE \gets R_\cE$, and $r_{\cP} \gets R_{\cP}$.
	\item Let $\cP_0$ be the algorithm that outputs $\stm=(\com,S)$ as a statement and $\cP=(\cP_0,\cP_1)$, where $\cP_1$'s randomness is fixed to $r_{\cP}$. Run $\cE_{\cP}(\stm;r_\cE)$ to extract the witness $(h,\rho)$. \label{step:extract-first}
	\item If $f(h,S)\neq 1$ or $\com\neq\commit(h;\rho)$, abort. 
	\item Repeat the following process: \label{step:repeat}
	\begin{enumerate}
		\item Sample $S' \sim \mathcal{D}^n$ and $r'_\cE \gets R_\cE$.
		\item Let $\cP'_0$ be the algorithm that outputs $\stm'=(\com,S')$ as a statement and $\cP'=(\cP'_0,\cP_1)$, where $\cP_1$'s randomness is fixed to $r_{\cP}$. Run $\cE_{\cP'}(\stm';r'_\cE)$ to extract the witness $(h',\rho')$.
		\item If $f(h',S')= 1$ and $\com=\commit(h';\rho')$, terminate and output $(h, \rho, h', \rho')$
		\item Else, go to step (a).
	\end{enumerate}
\end{enumerate}

\noindent\textbf{Running time:} 
Let $T$ be the random variable counting the number of calls to the inner extractor $\cE$ until termination.
For each fixed $\com$ and prover's randomness $i \in R_{\cP}$ we denote by $\epsilon_i$ the probability that $\cE_{\cP}((\com,S);r_{\cE})$ successfully outputs $(h,\rho)$ with $f(h,S)=1$ and $\com=\commit(h;\rho)$, where the probability is taken over $r_{\cE}\gets R_\cE$ and $S\sim\mathcal{D}^n$.
Denote by $E$ the event that $\cE_{\cP}$ successfully outputs a valid witness at Step~\ref{step:extract-first}.
We now evaluate the expected running time of $\ext$ as follows:

\begin{align*}
	\mathbb{E}[T] &= \sum_{i\in R_{\cP}} \big(\mathbb{E}[T \,|\, r_{\cP} = i \land E]\cdot \Pr[r_{\cP}=i \land E] \\
	&\qquad + \mathbb{E}[T \,|\, r_{\cP} = i \land \lnot E]\cdot \Pr[r_{\cP}=i \land \lnot E]\big)  \\
	&= \frac{1}{|R_{\cP}|} \sum_{i\in R_{\cP}} \big(\mathbb{E}[T \,|\, r_{\cP} = i \land E]\cdot \Pr[E \,|\, r_{\cP}=i] \\
	&\qquad + \mathbb{E}[T \,|\, r_{\cP} = i \land \lnot E]\cdot \Pr[\lnot E \,|\,r_{\cP}=i ]\big) \\
	&= \frac{1}{|R_{\cP}|} \sum_{i\in R_{\cP}} \big(\mathbb{E}[T \,|\, r_{\cP} = i \land E]\cdot \epsilon_i \\
	&\qquad + \mathbb{E}[T \,|\, r_{\cP} = i \land \lnot E]\cdot (1-\epsilon_i)\big) \\
	&\leq \frac{1}{|R_{\cP}|} \sum_{i\in R_{\cP}} \big( \tfrac{1}{\epsilon_i} \cdot \epsilon_i + 1\cdot (1- \epsilon_i)\big) \\
	&= 2
\end{align*}
where we used the fact that $\mathbb{E}[T \,|\, r_{\cP} = i \land \lnot E]=1$ since the algorithm terminates after the initial extraction fails.
Moreover, since the underlying $\ZKP$ is knowledge sound, each call to $\cE$ runs in expected polynomial time.
Overall, we conclude that $\ext$ runs in expected polynomial time.

\medskip
\noindent\textbf{Knowledge error:}
We define the following events:
\begin{itemize}
	\item $E$: $\cE_{\cP}(\stm;r_\cE)$ succeeds at Step~\ref{step:extract-first}
	\item $E'$: $\ext_{\prove^*}$ outputs valid $(h,\rho,h',\rho')$ 
	\item $E_1$: $E \land E'$
	\item $E_2$: $h = h'$
	\item $E_3$: $\tilde{F}(h) = 1$
\end{itemize}
Our goal is to relate the success probability $\pext$ of the extractor to $\pacc$, where 
\begin{align*}
	\pext &\geq \Pr[E_1 \land (\lnot E_2 \lor E_3)] \\
	\pacc &= \probrv{b=1}{(\com,b)\gets\interact{\prove^*}{\audit}}	
\end{align*}

To this end, we first rewrite (the lower bound of) $\pext$ as follows:
\begin{align*}
	\pext &\geq \Pr[E_1 \land (\lnot E_2 \lor E_3)] = \Pr[E_1] - \Pr[E_1 \land E_2 \land \lnot E_3]
\end{align*}
We now bound $\Pr[E_1]$ and $\Pr[E_1 \land E_2 \land \lnot E_3]$ separately.

\noindent\textbf{Bounding $\Pr[E_1]$:}
We rewrite $\Pr[E_1]$ as follows:
\begin{align*}
	\Pr[E_1] = \Pr[E'\,|\,E]\cdot\Pr[E]  \geq \pacc - \kappa
\end{align*}
where $\Pr[E'\,|\,E]=1$ since if $E$ happens, then $\epsilon_i\neq 0$ and the extractor always terminates in expected polynomial time and outputs $(h,\rho,h',\rho')$ at Step~\ref{step:repeat}.
The last inequality follows from the definition of $\pacc$ and the knowledge soundness of $\ZKP$.

\newcommand{\good}{\mathcal{G}}
\medskip
\noindent\textbf{Bounding $\Pr[E_1 \land E_2 \land \lnot E_3]$:}
For each fixed $\com$ and $i\in R_{\cP}$, let $\good_i$ be the set of ``good'' pairs $(r_{\cE},S)$ that cause $\cE_{\cP}((\com,S);r_{\cE})$ to output a valid witness $(h,\rho)$ such that $f(h,S)=1$ and $\com=\commit(h;\rho)$.
Then we have that 
\begin{align*}
\epsilon_i=\Pr_{(r_{\cE},S)\gets R_{\cE} \times \mathcal{D}^n}[(r_{\cE},S)\in\good_i]=\Pr[E\,|\, r_{\cP}=i] \\
=\Pr_{_{(r'_{\cE},S')\gets R_{\cE} \times \mathcal{D}^n}}\left[(r'_\cE,S')\in\good_i \,|\, r_{\cP}=i \land E\right].
\end{align*}
Given this, we rewrite $\Pr[E_1 \land E_2 \land \lnot E_3]$ as follows:
\begin{align*}
	& \Pr[E_1 \land E_2 \land \lnot E_3] \\
	&= \sum_{i\in R_{\cP}} \Pr[r_{\cP}=i]\cdot \Pr[E \land E' \land E_2 \land \lnot E_3 | r_{\cP}=i]  \\
	&= \frac{1}{|R_{\cP}|}\sum_{i\in R_{\cP}} \Pr[E | r_{\cP}=i] \cdot\Pr[E' \land E_2 \land \lnot E_3 | r_{\cP}=i \land E ]  \\
	&= \frac{1}{|R_{\cP}|}\sum_{i\in R_{\cP}} \epsilon_i \cdot\Pr[E' \land E_2 \land \lnot E_3 | r_{\cP}=i \land E ]  \\
	&\leq \frac{1}{|R_{\cP}|}\sum_{i\in R_{\cP}} \epsilon_i \cdot\probcond{f(h,S')=1  \land \tilde{F}(h) \neq 1}{r_{\cP}=i  \land E \\ \land (r'_\cE,S')\in\good_i}  \\
	&\leq \frac{1}{|R_{\cP}|}\sum_{i\in R_{\cP}} \epsilon_i \cdot \frac{\probcond{f(h,S')=1  \land \tilde{F}(h) \neq 1}{r_{\cP}=i  \land E}}{\probcond{(r'_\cE,S')\in\good_i}{r_{\cP}=i \land E}}  \\
	&= \frac{1}{|R_{\cP}|}\sum_{i\in R_{\cP}} \probcond{f(h,S')=1  \land \tilde{F}(h) \neq 1}{r_{\cP}=i  \land E}  \\
	&\leq \frac{1}{|R_{\cP}|}\sum_{i\in R_{\cP}} \probcond{f(h,S')=1}{r_{\cP}=i  \land E \land \tilde{F}(h) \neq 1}  \\
	%
	&\leq \pfpr
\end{align*}
where the last inequality follows from the assumption on $f$ and $\tilde{F}$, as at this stage $h$ is fixed by event $E$ and elements in each element of $S'$ is sampled i.i.d from $\mathcal{D}^n$.

Combining the bounds, we get
\begin{align*}
	\pext &\geq \Pr[E_1] - \Pr[E_1 \land E_2 \land \lnot E_3] \\
	&\geq \pacc - \kappa - \pfpr
\end{align*}

This completes the proof.

\noindent\underline{Zero Knowledge:}
We construct the following simulator, internally using the simulator $\cS$ for $\ZKP$: 

\noindent\underline{$\simul(1^\secp)$}
\begin{enumerate}
	\item Generate a dummy commitment $\com \gets \commit(0;\rho)$ using a uniformly random string $\rho\in\{0,1\}^{\ell_\rho}$.
	\item Sample $S_\text{audit} \sim \mathcal{D}^n$.
	\item Run $\cS((\com,S_\text{audit}))$ to get a simulated $\view'$ for $\ZKP$.
	\item Output $\view = (\com,S_\text{audit},\view')$.
\end{enumerate}
Since $\commit$ is hiding, the dummy commitment is indistinguishable from a real commitment.
Moreover, the output of $\cS$ is indistinguishable from the view of $\audit$ during the interaction $\interact{\cP(h)}{\cV}((\com, S_\text{audit}))$ for any valid witness $h$.
Since a real execution of $\auditcsp$ defines $\stm = (\com, S_\text{audit})$ and $\wit = h$ such that $\rel(\stm,\wit)=1$ except with probability at most $\pfnr$, by setting $\pfnr$ to be negligible in $\secp$, the output of $\simul$ is indistinguishable from the view of $\audit$ during the interaction $\interact{\prove(h)}{\audit}$.
This completes the proof.
\end{proof}

\subsubsection{Example Instantiation: $\auditcsp$ for Accuracy Auditing}\label{sec:auditcsp-accuracy}
To instantiate $\auditcsp$ for accuracy auditing, we consider the empirical and true error rates as follows:
\begin{align*}
	\hat{\ell}_S(h) &= \frac{1}{n}\sum_{(x,y)\in S}\mathbb{I}(h(x)\neq y)\\
	\ell(h) &= \mathbb{E}_{(x,y)\sim \mathcal{D}}[\mathbb{I}(h(x)\neq y)] 
\end{align*}
where $n = |S|$, and define the empirical predicate $f$, the distributional predicate $F$, and the relaxed predicate $\tilde{F}$ as follows.
\begin{align*}
    f(h,S_\text{audit}) = 1  & \iff \hat{\ell}_S(h) \leq t+\delta \\
    F(h) = 1 & \iff \ell(h) \leq t  \\
	\tilde{F}(h) = 1 & \iff \ell(h) \leq t + 2\delta
\end{align*}

To apply Theorem~\ref{thm:auditcsp-general} to accuracy auditing, it would be sufficient to find the false negative rate $\pfnr$ and false positive rate $\pfpr$ by the following lemma, and then set $n\delta^2\in\Omega(\secp)$.
\begin{lemma}\label{lem:fn-fp-accuracy}
For any hypothesis $h$,
\begin{align}
	& \Pr_{S_\text{audit}\sim \mathcal{D}^n}[f(h,S_\text{audit})\neq 1 \mid F(h)=1] \leq 2 e^{-2n\delta^2} \label{eq:false-negative-accuracy} \\
	& \Pr_{S_\text{audit}\sim \mathcal{D}^n}[f(h,S_\text{audit})=1 \mid \tilde{F}(h) \neq 1] \leq 2 e^{-2n \delta^2} \label{eq:false-positive-accuracy}
\end{align}
\end{lemma}
\begin{proof}
Consider the empirical error $\hat{\ell}_S(h) = \frac{1}{n}\sum_{(x,y)\in S}\mathbb{I}(h(x)\neq y)$ and the true error $\ell(h) = \mathbb{E}_{(x,y)\sim \mathcal{D}}[\mathbb{I}(h(x)\neq y)]$.
Since each element of $\setaudit$ is sampled i.i.d. from $\mathcal{D}$, by Hoeffding's inequality, we have
\begin{align*}
	\Pr_{S_\text{audit}\gets \mathcal{D}^n}[|\hat{\ell}_S(h) - \ell(h)| \geq \delta] \leq 2e^{-2n\delta^2}.
\end{align*}
Thus, for any $h$ such that $F(h)=1$ (i.e., $\ell(h) \leq t$), the probability that $f(h,S_\text{audit})\neq 1$ (i.e., $\hat{\ell}_S(h) > t + \delta$) is at most $2e^{-2n\delta^2}$.

Similarly, for any $h$ such that $\tilde{F}(h)\neq 1$ (i.e., $\ell(h)  > t+2\delta$), the probability that $f(h,S_\text{audit})=1$ (i.e., $\hat{\ell}_S(h) \leq t + \delta$) is at most $2 e^{-2n\delta^2}$.
\end{proof}

Overall, by Theorem~\ref{thm:auditcsp-general} and Lemma~\ref{lem:fn-fp-accuracy}, we conclude that $\auditcsp$ instantiated for accuracy auditing is a secure auditing protocol for accuracy with the following properties for a sufficiently large $n\delta^2\in\Omega(\secp)$:

\begin{itemize}
    \item If the true error of the model $h$ is $\leq t$, then the auditor accepts with high probability.
    \item If the auditor accepts, then the auditor gets assurance that the true error of the model $h$ is at most $t + 2\delta$, even if the prover is misbehaving.
\end{itemize}

\subsubsection{Example Instantiation: $\auditcsp$ for Demographic Parity Auditing}\label{sec:auditcsp-demographic-parity}
Similarly, we can instantiate the protocol $\auditcsp$ for auditing fairness conditions.
Take the demographic parity as an example, for which we can set the query space to $\querysp = \{x_i\}^m_{i=1}$ without ground-truth labels.
We consider the empirical and true demographic parity differences as follows: 
\begin{align*}
	\dpemp(h,\setaudit) &= \left|\frac{1}{n_0}\sum_{x\in S_0}\mathbb{I}(h(x) = 1) - \frac{1}{n_1}\sum_{x\in S_1}\mathbb{I}(h(x) = 1) \right| \\ 
	\dptrue(h) &=  |\mathrm{E}_{x \sim \mathcal{D}}[\mathbb{I}(h(x) =1) \,|\, s_x = 0] \\
	&\qquad - \mathrm{E}_{x \sim \mathcal{D}}[\mathbb{I}(h(x) =1) \,|\, s_x = 1]| 
\end{align*}
where $s_x$ denotes the sensitive feature of a data point $x$, $S_0 = \{x \in \setaudit : s_x = 0 \}$, $S_1 = \{x \in \setaudit : s_x = 1 \}$, $n_0 = |S_0|$, and $n_1 = |S_1|$.

Define the empirical predicate $f$, the distributional predicate $F$, and the relaxed predicate $\tilde{F}$ as follows.
\begin{align*}
    f(h,S_\text{audit}) = 1  & \iff \dpemp(h,\setaudit) \leq t+2\delta \\
    F(h) = 1 & \iff \dptrue(h) \leq t  \\
	\tilde{F}(h) = 1 & \iff \dptrue(h) \leq t + 4\delta
\end{align*}

To apply Theorem~\ref{thm:auditcsp-general} to demographic parity auditing, we can prove the following lemma in place of Lemma \ref{lem:fn-fp-accuracy}, and then set $n_{\min}\delta^2\in\Omega(\secp)$.

\begin{lemma}\label{lem:false-positive-demographic-parity}
For any hypothesis $h$,
\begin{align}
	& \Pr_{S_\text{audit}\sim \mathcal{D}^n}[f(h,S_\text{audit})\neq 1 \mid F(h)=1] \leq 4 e^{-2n_{\min}\delta^2} \label{eq:false-negative-demographic-parity} \\
	& \Pr_{S_\text{audit}\sim \mathcal{D}^n}[f(h,S_\text{audit})=1 \mid \tilde{F}(h) \neq 1] \leq 4 e^{-2n_{\min}\delta^2} \label{eq:false-positive-demographic-parity}
\end{align}
where $n_{\min} = \min(n_0,n_1)$.
\end{lemma}
\begin{proof}	
We first prove \eqref{eq:false-positive-demographic-parity}. Define the following variables: 
\begin{align*}
	& g_0 = \frac{1}{n_0}\sum_{x\in S_0}\mathbb{I}(h(x) = 1) 
	& p_0 = \mathrm{E}_{x \sim \mathcal{D}}[\mathbb{I}(h(x) =1) \,|\, s_x = 0]
	\\
	& g_1 = \frac{1}{n_1}\sum_{x\in S_1}\mathbb{I}(h(x) = 1)
	& p_1 = \mathrm{E}_{x \sim \mathcal{D}}[\mathbb{I}(h(x) =1) \,|\, s_x = 1] 
\end{align*}
By Hoeffding's inequality, we have
\begin{align*}
	& \Pr[|g_0 - p_0| \geq \delta] \leq 2 e^{-2n_0\delta^2} \\
	& \Pr[|g_1 - p_1| \geq \delta] \leq 2 e^{-2n_1\delta^2}
\end{align*}
where the probability is taken over the randomness of $\setaudit \sim \mathcal{D}^n$.
Thus, for any $h$ such that $\tilde{F}(h)\neq 1$ (i.e., $|p_0 - p_1| > t+4\delta$), the probability that $f(h,S_\text{audit})=1$ (i.e., $|g_0 - g_1| \leq t + 2\delta$) can be bounded as follows: 

\begin{align*}
	\Pr[|g_0-g_1| \leq t + 2\delta] &\leq \Pr[|g_0-g_1| \leq |p_0 - p_1|-2\delta]  \\
	&\leq \Pr[2\delta \leq |p_0-p_1 - (g_0 - g_1)| ] \\
	&\leq \Pr[2\delta \leq |p_0-g_0| + |p_1 - g_1| ] \\
	&\leq \Pr[\delta \leq |p_0-g_0| \lor  \delta \leq |p_1 - g_1| ] \\
	&\leq \Pr[\delta \leq |p_0-g_0|] + \Pr[\delta \leq |p_1 - g_1| ]  \\
	&\leq 2 e^{-2n_0\delta^2} + 2 e^{-2n_1\delta^2} \leq 4 e^{-2n_{\min}\delta^2}
\end{align*}

Analogously, we can prove \eqref{eq:false-negative-demographic-parity}.
For any $h$ such that ${F}(h) = 1$ (i.e., $|p_0 - p_1| \leq t$), the probability that $f(h,S_\text{audit})\neq 1$ (i.e., $|g_0 - g_1| > t + 2\delta$) can be bounded in the same way:
\begin{align*}
	 &\Pr[|g_0-g_1| > t + 2\delta] \\
	&\leq \Pr[|g_0-g_1| > |p_0 - p_1|+2\delta] \leq 4 e^{-2n_{\min}\delta^2}
\end{align*}
\end{proof}

Overall, by Theorem~\ref{thm:auditcsp-general} and Lemma~\ref{lem:false-positive-demographic-parity}, we conclude that $\auditcsp$ instantiated for demographic parity auditing is a secure auditing protocol for demographic parity with the following properties for a sufficiently large $n_{\min}\delta^2\in\Omega(\secp)$:

\begin{itemize}
    \item If the true demographic parity of the model $h$ is $\leq t$, then the auditor accepts with high probability.
    \item If the auditor accepts, then the auditor gets assurance that the true demographic parity of the model $h$ is at most $t + 4\delta$, even if the prover is misbehaving.
\end{itemize}

\fi
\section{Further Evaluation}\label{sec:eval-appendix}

First, we present in Figure~\ref{fig:accuracy-attack-additional} results for attacking accuracy audits on additional datasets mentioned in \S\ref{sec:eval}.

\begin{figure}[t]
    \centering
    \includegraphics[width=\linewidth]{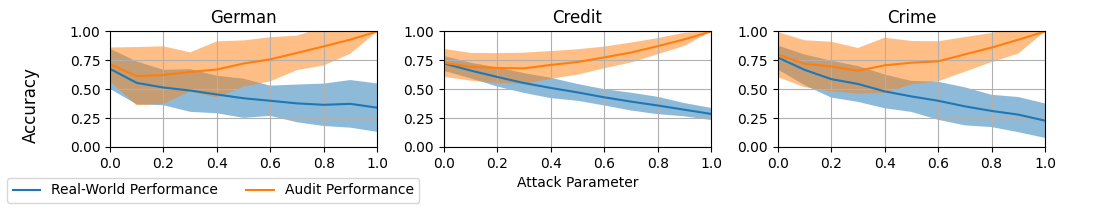}
    \caption{Accuracy of models trained on datasets constructed by Algorithm \ref{alg:attack} on various benchmarks. Values are averages over ten runs, error bars represent one standard deviation. }
    \label{fig:accuracy-attack-additional}
\end{figure}

Next, we present the application of the attack described in \S~\ref{sec:attack-methods}. In this attack, the adversary is attempting to maximize the model's denial rate $\Pr_{x\sim \mathcal{D}}[h(x)=0]$ while still appearing accurate to the audit. The results of this attack are given in Figure~\ref{fig:denial}. Observe that as the attack parameter approaches $1$ (and the attack becomes maximally malicious), the denial rate of the model on the audit set remains close to the fully honest denial rate while the denial rate on independently sampled data approaches $1$. Similarly, the accuracy of the model on the audit set approaches $1$, while the true accuracy decreases down to roughly $0.6$ (this reflects the true denial rate of the distribution).

\begin{figure}[t]
    \centering
    \includegraphics[width=0.5\linewidth]{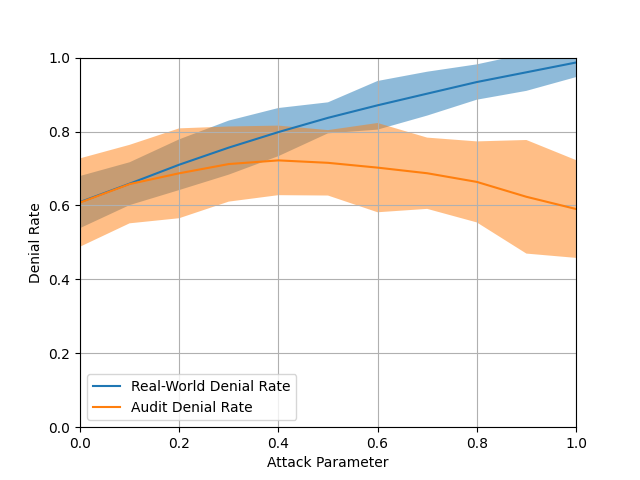}\includegraphics[width=0.5\linewidth]{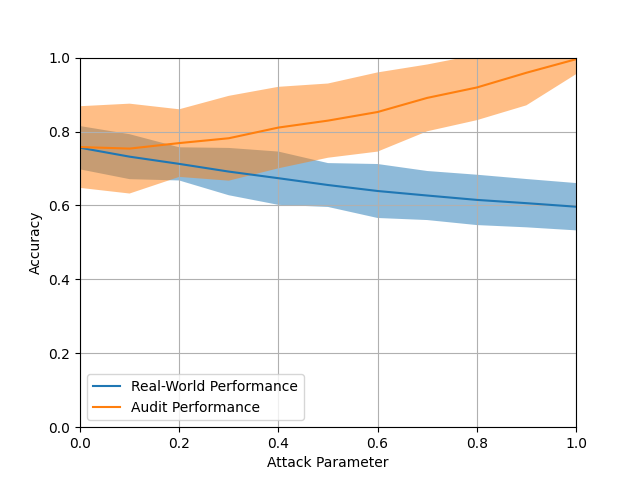}
    \caption{Accuracy and denial rates of models trained on datasets constructed by Algorithm \ref{alg:attack} on ACSEmployment. Values are averages over ten runs, error bars represent one standard deviation.}
    \label{fig:denial}
\end{figure}

\subsection{Attacking XGBoost}\label{sec:xgboost}

We next show that a minor modification of Algorithm~\ref{alg:attack} suffices to evade accuracy auditing of XGBoost models~\cite{chen2016xgboost}.
Unlike a single decision tree, XGBoost trains a gradient-boosted ensemble of trees, so the theoretical guarantee of Theorem~\ref{thm:attack-isolation} does not apply: with only one copy of $S_{\mathit{audit}}$ in $S'_{\mathit{train}}$, the boosting process assigns excessive weight to the $2d$ flipped-label neighborhood points, causing the ensemble to misclassify the audit points themselves and thereby fail the audit.
To counteract this, we include $2d$ copies of $S_{\mathit{audit}}$ (with true labels) in $S'_{\mathit{train}}$ rather than a single copy, so that the correct-label signal at each audit point outweighs its $2d$ flipped-label neighbors; all other details of the construction remain unchanged.
We used the perturbation magnitude $\varepsilon = 1.0$ throughout.
The XGBoost ensemble is configured with $500$ trees of maximum depth $10$, learning rate $\eta = 0.3$, L2 regularization $\lambda = 1.0$, and minimum split-loss $\gamma = 0$; all other XGBoost hyperparameters are left at their defaults.
We empirically evaluated this attack on the Default Credit, Adult Income, and ACS Income datasets; results are shown in Figure~\ref{fig:xgboost-accuracy}.
Our modified attack successfully evades accuracy auditing across all benchmarks.

\begin{figure}[t]
    \centering
    \includegraphics[width=\linewidth]{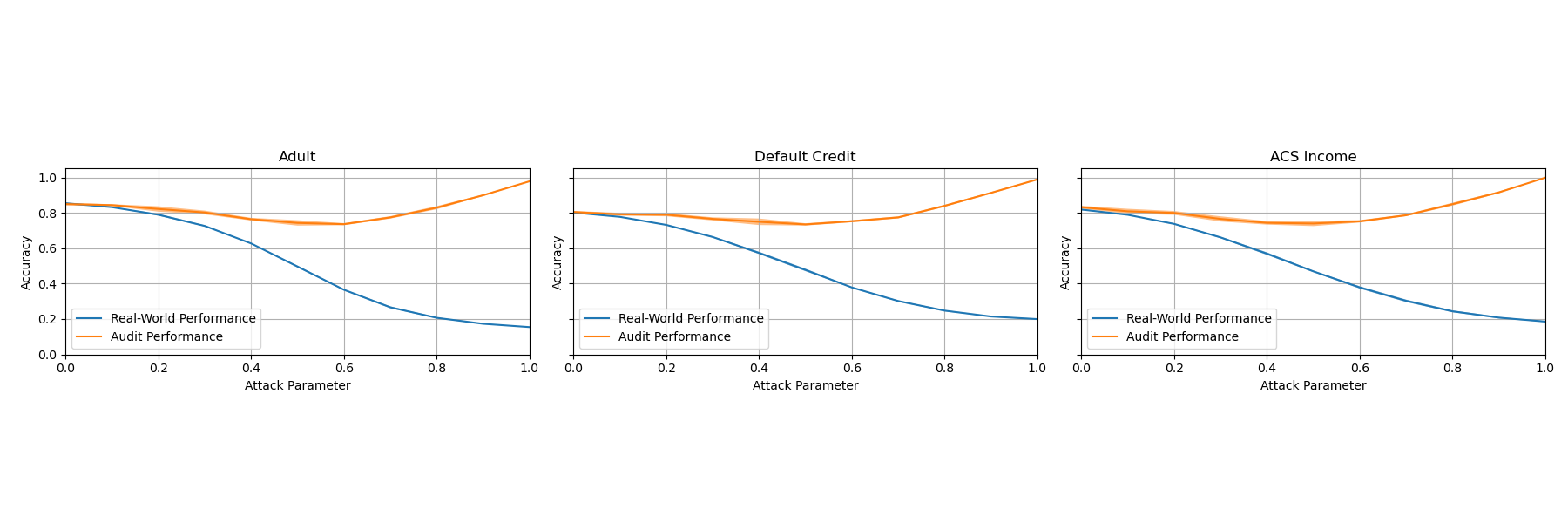}
    \caption{Accuracy of XGBoost models trained on inflated datasets constructed by the modified Algorithm~\ref{alg:attack} on Adult, Default Credit, and ACS Income datasets. Values are averages over five runs; error bars represent one standard deviation.}
    \label{fig:xgboost-accuracy}
\end{figure}

\subsection{Attacking Neural Networks}\label{sec:neural}
Finally, we present an evaluation of a modified version of the attack that targets neural networks rather than decision trees. Whereas decision trees have very specific conditions that allow us to constrain their behavior, it is much harder to provide theoretical guarantees for neural networks. In order to encourage memorization of the training data, we used a relatively shallow network with very large individual layers. Our attack samples a large amount of training data, and decides whether to label each point with the honest label or dishonest label depending on its proximity to the nearest audit data point. We evaluated this attack on an 8-dimensional mixture of Gaussian distributions; the results are shown in Figure~\ref{fig:neuralnet}.

We also performed this attack on ACS Employment data, and over 10 runs, trained a model with average accuracy of 69.83\% accuracy on held-out test data while maintaining an average of 88.82\% accuracy on audit data.

\begin{figure}[t]
    \centering
    \includegraphics[width=0.8\linewidth]{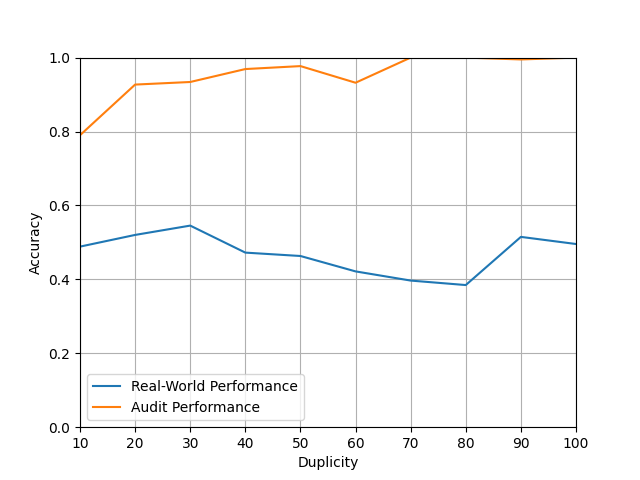}
    \caption{Performance of 226M-parameter neural networks trained on datasets constructed from a mixture of Gaussian distributions. Duplicity refers to the number of perturbed copies of the audit dataset included in the training data.}
    \label{fig:neuralnet}
\end{figure}

\subsection{Evading Statistical Methods for Attack Detection}
\iffull

We empirically examine our ability to observe data forging attacks by applying statistical tests to the datasets, as described in \S~\ref{sec:detection}. There is no singular way to determine whether two samples were drawn from the same distribution, so we apply some common statistical tools. In particular, our goal is to determine if the distribution from which the audit data is drawn is identical to the distribution from which the training data is drawn. We use Welch's $t$-test, which serves to determine whether two distributions have the same mean, and Levene's test, a one-way ANOVA for determining whether two distributions have the same variance. These tests are typically applied to 1-dimensional data, and so we apply them to each feature individually. The results of these experiments are given in Table~\ref{tab:stats}.
\begin{table*}
\centering
\caption{Summary and Test statistics for Age feature on ACSEmployment, conditioned on label. Test statistics used are Welch's $t$-test and Levene's test. Attack is undetectable when summary statistics are similar to honest ones, and when test statistics are close to 0. Comparisons are between fully honest and fully malicious datasets.} 
\begin{tabular}{|c|c|c|c|c|c|}
    \hline
    \multicolumn{2}{|c|}{\multirow{2}{*}{Age}} & \multicolumn{2}{|c|}{Label $=0$} & \multicolumn{2}{|c|}{Label $=1$} \\ \cline{3-6}
    \multicolumn{2}{|c|}{} & Honest & Attack & Honest & Attack \\ \hline
    Summary & $\mu$ & 41.6651 & 41.9657 & 43.9184 & 43.8131 \\ \cline{2-6}
    Statistics & $\sigma^2$ & 804.5804 & 810.8822 & 223.1269 & 221.42394 \\ \hline
    Test & $t$-test & 0.6521 & 0.0033 & 0.7067 & 0.0110 \\ \cline{2-6}
    Statistics & ANOVA & 0.6200 & 0.0026 & 1.6500 & 0.0186 \\ \hline
    \hline
    \multicolumn{2}{|c|}{\multirow{2}{*}{Education}} & \multicolumn{2}{|c|}{Label $=0$} & \multicolumn{2}{|c|}{Label $=1$} \\ \cline{3-6}
    \multicolumn{2}{|c|}{} & Honest & Attack & Honest & Attack \\ \hline
    Summary & $\mu$ & 13.39761692 & 13.41700338 & 18.45539675 & 18.50545506 \\ \cline{2-6}
    Statistics & $\sigma^2$ & 42.99789908 & 42.16899485 & 9.979327135 & 8.943082831 \\ \hline
    Test & $t$-test & 0.7984001575 & 0.0356390553 & 0.9499974697 & 0.1302788154 \\ \cline{2-6}
    Statistics & ANOVA & 0.4844657261 & 0.0003374130653 & 1.227829625 & 0.02531893152 \\ \hline
    \hline
    \multicolumn{2}{|c|}{\multirow{2}{*}{Military Status}} & \multicolumn{2}{|c|}{Label $=0$} & \multicolumn{2}{|c|}{Label $=1$} \\ \cline{3-6}
    \multicolumn{2}{|c|}{} & Honest & Attack & Honest & Attack \\ \hline
    Summary & $\mu$ & 2.5794 & 2.5834 & 3.8121 & 3.8302 \\ \cline{2-6}
    Statistics & $\sigma^2$ & 3.2749 & 3.2648 & 0.3507 & 0.3265 \\ \hline
    Test & $t$-test & 0.4997 & 0.0313 & 0.8699 & 0.1755 \\ \cline{2-6}
    Statistics & ANOVA & 1.0240 & 0.0009 & 1.2394 & 0.0304 \\ \hline
\end{tabular}
\label{tab:stats}
\end{table*}

We observe that the summary statistics of the malicious training data closely match the values for the honest data, suggesting that comparing these two values would not be a successful detection mechanism. This is compounded by the fact that the test statistics for Welch's $t$-test and Levene's test for the malicious training data are considerably smaller on average than the same test statistics for the honest training data, corroborating higher rate of passing the hypothesis tests we observe. At a significance level of $\alpha=0.05$, we expect a false positive rate of approximately $5\%$. On the other hand, we observe a $0\%$ true positive rate. We note that in a practical application of this attack, the auditor would have access only to the honest or malicious values over a single training run, and would thus be unable to easily distinguish between the two cases by comparing the values or by looking at averages over many runs as we have done here. That being said, an auditor may find it suspicious if the p-value returned by a statistical test is extremely low (even though such a scenario may be very plausible for some distributions); an attacker can safely relax this attack to a comfortable degree, though doing so will increase the risk of failing the audit.\fi

\ifneurips
\clearpage
\newpage
\section*{NeurIPS Paper Checklist}

\begin{enumerate}

\item {\bf Claims}
    \item[] Question: Do the main claims made in the abstract and introduction accurately reflect the paper's contributions and scope?
    \item[] Answer: \answerYes{} 
    \item[] Justification: Our main technical results are in \S\ref{sec:attack}.

\item {\bf Limitations}
    \item[] Question: Does the paper discuss the limitations of the work performed by the authors?
    \item[] Answer: \answerYes{} 
    \item[] Justification: We elaborate on limitations of the undetectability of our attacks in \S\ref{sec:detection}.

\item {\bf Theory assumptions and proofs}
    \item[] Question: For each theoretical result, does the paper provide the full set of assumptions and a complete (and correct) proof?
    \item[] Answer: \answerYes{} 
    \item[] Justification: We give a set of proofs in \S\ref{sec:appendix}.

    \item {\bf Experimental result reproducibility}
    \item[] Question: Does the paper fully disclose all the information needed to reproduce the main experimental results of the paper to the extent that it affects the main claims and/or conclusions of the paper (regardless of whether the code and data are provided or not)?
    \item[] Answer: \answerYes{} 
    \item[] Justification: We explain how we performed the experiments in \S\ref{sec:eval}.

\item {\bf Open access to data and code}
    \item[] Question: Does the paper provide open access to the data and code, with sufficient instructions to faithfully reproduce the main experimental results, as described in supplemental material?
    \item[] Answer: \answerNo{} 
    \item[] Justification: We will provide open access to the data and code in the future.

\item {\bf Experimental setting/details}
    \item[] Question: Does the paper specify all the training and test details (e.g., data splits, hyperparameters, how they were chosen, type of optimizer, etc.) necessary to understand the results?
    \item[] Answer: \answerYes{} 
    \item[] Justification: We explain how we performed the experiments in \S\ref{sec:eval}.

\item {\bf Experiment statistical significance}
    \item[] Question: Does the paper report error bars suitably and correctly defined or other appropriate information about the statistical significance of the experiments?
    \item[] Answer: \answerNo{} 
    \item[] Justification: 

\item {\bf Experiments compute resources}
    \item[] Question: For each experiment, does the paper provide sufficient information on the computer resources (type of compute workers, memory, time of execution) needed to reproduce the experiments?
    \item[] Answer: \answerYes{} 
    \item[] Justification: We explain how we performed the experiments in \S\ref{sec:eval}.
    
\item {\bf Code of ethics}
    \item[] Question: Does the research conducted in the paper conform, in every respect, with the NeurIPS Code of Ethics \url{https://neurips.cc/public/EthicsGuidelines}?
    \item[] Answer: \answerYes{} 
    \item[] Justification: 

\item {\bf Broader impacts}
    \item[] Question: Does the paper discuss both potential positive societal impacts and negative societal impacts of the work performed?
    \item[] Answer: \answerNA{} 
    \item[] Justification: 
    
\item {\bf Safeguards}
    \item[] Question: Does the paper describe safeguards that have been put in place for responsible release of data or models that have a high risk for misuse (e.g., pretrained language models, image generators, or scraped datasets)?
    \item[] Answer: \answerNA{} 
    \item[] Justification: 

\item {\bf Licenses for existing assets}
    \item[] Question: Are the creators or original owners of assets (e.g., code, data, models), used in the paper, properly credited and are the license and terms of use explicitly mentioned and properly respected?
    \item[] Answer: \answerNA{} 
    \item[] Justification: 

\item {\bf New assets}
    \item[] Question: Are new assets introduced in the paper well documented and is the documentation provided alongside the assets?
    \item[] Answer: \answerNA{} 
    \item[] Justification: 

\item {\bf Crowdsourcing and research with human subjects}
    \item[] Question: For crowdsourcing experiments and research with human subjects, does the paper include the full text of instructions given to participants and screenshots, if applicable, as well as details about compensation (if any)? 
    \item[] Answer: \answerNA{} 
    \item[] Justification: 

\item {\bf Institutional review board (IRB) approvals or equivalent for research with human subjects}
    \item[] Question: Does the paper describe potential risks incurred by study participants, whether such risks were disclosed to the subjects, and whether Institutional Review Board (IRB) approvals (or an equivalent approval/review based on the requirements of your country or institution) were obtained?
    \item[] Answer: \answerNA{} 
    \item[] Justification: 

\item {\bf Declaration of LLM usage}
    \item[] Question: Does the paper describe the usage of LLMs if it is an important, original, or non-standard component of the core methods in this research? Note that if the LLM is used only for writing, editing, or formatting purposes and does not impact the core methodology, scientific rigorousness, or originality of the research, declaration is not required.
    \item[] Answer: \answerYes{} 
    \item[] Justification: The authors used LLMs to polish writing, and for assistance with literature search in some components of this paper. We also used generative AI to create some of the icons in Figure~\ref{fig:auditprocess}. In addition, we used an LLM for assistance with Lemma~\ref{lem:t-dist-bound}. We checked the proof assistance thoroughly by hand before including it in the paper.

\end{enumerate}
\fi


\end{document}